\documentclass[a4paper,UKenglish,cleveref, autoref, thm-restate]{lipics-v2021}
\usepackage{amsmath, amssymb, amsfonts}
\usepackage{todonotes}
\usepackage{amsmath,amsfonts,amssymb,amsthm}
\usepackage{graphicx,color}
\usepackage{boxedminipage}
\usepackage{algpseudocode}
\usepackage{framed}
\usepackage{xspace}
\usepackage{todonotes}
\usepackage{footnote}
\usepackage{ifthen}
\newboolean{shortver}
\usepackage{nicefrac}

\setboolean{shortver}{false}
\usepackage{mathtools}
\usepackage{mathrsfs}
\usepackage{todonotes}
\usepackage{footnote}
\usepackage{rotating}
\usepackage{adjustbox}

\usepackage{thmtools}
\usepackage{thm-restate}

\newcommand{\cR}{{\mathcal{R}}}

\newcommand{\bR}{{\mathbb{R}}}
\usepackage{mdframed}

\newcommand{\omitted}{{\ifthenelse{\boolean{shortver}}{$\spadesuit$}{\!}}}

\usepackage[most]{tcolorbox}

\newtcolorbox{problemBox}[2]{%
  enhanced,
  colback=white,
  boxrule=0.5pt,
  sharp corners,
  left=1.2mm,right=1.2mm,top=1.0mm,bottom=1.0mm,
  title={\color{black} #1 (#2)},
  attach boxed title to top center={xshift=.10cm, yshift*=-2.5mm},
  boxed title style={size=small, frame hidden, colback=white},
}

\tcbuselibrary{skins}

\usepackage{xspace}
\newcommand{\cO}{\mathcal{O}}
\newcommand{\no}{\texttt{No}\xspace} 
\newcommand{\yes}{\texttt{Yes}\xspace}

\makeatletter
\newcommand*{\rom}[1]{\expandafter\@slowromancap\romannumeral #1@}
\makeatother

\usepackage{mathtools}

\definecolor{anti-flashwhite}{rgb}{0.95, 0.95, 0.96}

\newtheorem{reduction rule}{Reduction Rule}
\newtheorem*{reduction rule*}{Reduction Rule}

 \usepackage{xifthen}
 \usepackage{tabularx}
\usepackage{tikz} 
\usetikzlibrary{arrows.meta,calc,positioning,decorations.pathreplacing}
 \usepackage{thm-autoref}
\usepackage{enumitem}

\newcommand{\Oh}{\mathcal{O}}

\newcommand{\mtcsp}{\textsc{Monotone $2$-CSP}\xspace}
\newcommand{\ddmtcsp}{\textsc{Distinct Domain Monotone $2$-CSP}\xspace}

\newcommand{\cbra}[1]{\left\{ #1 \right\}}
\newcommand{\bin}{\{0,1\}}

\newcommand{\val}{\mathsf{val}}

\newlength{\RoundedBoxWidth}
\newsavebox{\GrayRoundedBox}
\newenvironment{GrayBox}[1]%
   {\setlength{\RoundedBoxWidth}{.93\textwidth}
    \def\boxheading{#1}
    \begin{lrbox}{\GrayRoundedBox}
       \begin{minipage}{\RoundedBoxWidth}}%
   {   \end{minipage}
    \end{lrbox}
    \begin{center}
    \begin{tikzpicture}%
       \node(Text)[draw=black!20,fill=white,rounded corners,%
             inner sep=2ex,text width=\RoundedBoxWidth]%
             {\usebox{\GrayRoundedBox}};
        \coordinate(x) at (current bounding box.north west);
        \node [draw=white,rectangle,inner sep=3pt,anchor=north west,fill=white] 
        at ($(x)+(6pt,.75em)$) {\boxheading};
    \end{tikzpicture}
    \end{center}}     

\newenvironment{defproblemx}[2][]{\noindent\ignorespaces%
                                \FrameSep=6pt%
                                \parindent=0pt%
                \vspace*{-1.5em}
                \ifthenelse{\isempty{#1}}{%
                  \begin{GrayBox}{\textsc{#2}}%
                }{%
                  \begin{GrayBox}{\textsc{#2}  parameterized by~{#1}}%
                }
                \begin{tabular*}{\textwidth}{@{\hspace{.1em}} >{\itshape} p{1.8cm} p{0.8\textwidth} @{}}%
            }{
                \end{tabular*}%
                \end{GrayBox}%
                \ignorespacesafterend
            }  

\newcommand{\cI}{\mathcal{I}}

\usepackage{relsize}

\newtoggle{includeFootnote}
\toggletrue{includeFootnote}

\usepackage{soul}

\newcommand{\p}{\mathsf{P}}

\definecolor{newgray}{gray}{0.25} 

\makeatletter
\renewcommand{\paragraph}{%
  \@startsection{paragraph}{4}%
  {\z@}{1ex \@plus 1ex \@minus .2ex}{-0.5em}%
  {\normalfont\normalsize\bfseries}%
}
\makeatother 
\newcommand{\bd}{\textsf{bd}\xspace}
\newcommand{\impline}{{\sf ImpLines}\xspace}

\newcommand{\indx}{{\sf Ind}\xspace}

\newcommand{\hor}{{\sf hor}\xspace}
\newcommand{\ver}{{\sf vert}\xspace}
\newcommand{\gridpts}{{\sf Gridpts}\xspace}

\providecommand{\ptsend}{{\sf Endpts}\xspace}
\newcommand{\blkend}{{\sf blkend}\xspace}
\newcommand{\block}{{\sf block}\xspace}
\newcommand{\gapnum}{{\sf GapNum}\xspace}
\newcommand{\gapfn}{{\sf GapFn}\xspace}
\newcommand{\exppat}{{\sf ExpPat}\xspace}

\newcommand{\relz}{{\sf realize}\xspace}

\newcommand{\pcbvcfull}{\textsc{Constrained Bipartite Vertex Cover}\xspace}
\newcommand{\pcbvcshort}{\textsc{CBVC}\xspace}

\newcommand{\plscfull}{\textsc{L-Shape Cover}\xspace}
\newcommand{\plscshort}{\textsc{LSC}\xspace}

\newcommand{\pnelscfull}{\textsc{$\NE$-L-Shape Cover}\xspace}
\newcommand{\pnelscshort}{\textsc{$\NE$-LSC}\xspace}

\newcommand{\NE}{\mathrm{NE}\xspace}
\newcommand{\NW}{\mathrm{NW}\xspace}
\newcommand{\SE}{\mathrm{SE}\xspace}
\newcommand{\SW}{\mathrm{SW}\xspace}

\newcommand{\side}{\mathsf{s}} 

\newcommand{\suppline}{\operatorname{line}}

\newcommand{\si}{\diamond}
\newcommand{\cE}{\mathcal{E}}

\newcommand{\predop}{\operatorname{pred}}

\newcommand{\succX}{\mathsf{succ}_X}
\newcommand{\succY}{\mathsf{succ}_Y}

\newcommand{\coord}{\mathsf{coord}}

\newcommand{\prbcfull}{{\sc Boundary Covering with Continuous Axis-Parallel Rectangles}\xspace}
\newcommand{\prbcshort}{{\sc BCCAPR}\xspace} 

\newcommand{\bcdaprfull}{{\sc Boundary Covering with Discrete Axis-Parallel Rectangles}\xspace} 
\newcommand{\bcdaprshort}{{\sc BCDAPR}\xspace}

\newcommand{\prbcLfull}{{\sc Boundary Covering with Continuous Axis-Parallel L-shapes}\xspace}
\newcommand{\prbcLshort}{{\sc BCCAPL}\xspace} 

\newcommand{\rcsp}{\textsc{$3$-Regular $2$-CSP}\xspace}

\title{
Covering Points with Rectangular Boundaries} 

\titlerunning{Covering Points with Rectangular Boundaries} 

\author{Madhumita Kundu}{University of Bergen, Norway}{kundumadhumita.134@gmail.com}{}{}

\author{Daniel Lokshtanov}{University of California Santa Barbara, USA}{daniello@ucsb.edu}{}{}

\author{Soumi Nandi}{The Institute of Mathematical Sciences, India}{nandisoumi1@gmail.com}{https://orcid.org/0000-0002-1000-991X}{}

\author{Saket Saurabh}{The Institute of Mathematical Sciences, India  \and University of Bergen, Norway}{saket@imsc.res.in}{https://orcid.org/0000-0001-7847-6402}{}

\author{Kushal Singanporia}{The Institute of Mathematical Sciences, India}{kushal03132@gmail.com}{}{}

\authorrunning{M. Kundu et al.} 

\Copyright{Jane Open Access and Joan R. Public} 


\category{} 

\relatedversion{} 

\EventEditors{John Q. Open and Joan R. Access}
\EventNoEds{2}
\EventLongTitle{42nd Conference on Very Important Topics (CVIT 2016)}
\EventShortTitle{CVIT 2016}
\EventAcronym{CVIT}
\EventYear{2016}
\EventDate{December 24--27, 2016}
\EventLocation{Little Whinging, United Kingdom}
\EventLogo{}
\SeriesVolume{42}
\ArticleNo{23}

\begin{document}

\maketitle

\begin{abstract}
Geometric covering problems ask for a small family of geometric objects whose union covers a given point set. We study the more restrictive \emph{boundary covering} variant, where every point must lie on the boundary of a chosen object. Motivated by the framework of Langerman and Morin\,[Discret.\ Comput.\ Geom., 2005] for hyperspheres, we initiate the study of boundary covering by axis-parallel rectangles.

We first consider the \emph{discrete} setting, where rectangles must be selected from a given family. We define \bcdaprfull\ (\bcdaprshort): given a point set \(P\subseteq\mathbb{R}^2\), a family \(\mathcal{R}\) of axis-parallel rectangles, and an integer \(k\), decide whether \(P\) can be covered by the boundaries of at most \(k\) rectangles from \(\mathcal{R}\). We prove that \bcdaprshort\ is \(\mathrm{W}[1]\)-hard parameterized by \(k\).

We then study the \emph{continuous} variant, \prbcfull\ (\prbcshort), where rectangles may be placed freely. Given \(P\subseteq\mathbb{R}^2\) and \(k\), the goal is to decide whether \(P\) can be covered by the boundaries of at most \(k\) axis-parallel rectangles. In contrast to the discrete case, we show that \prbcshort\ is fixed-parameter tractable, with running time \(2^{\cO(k\log k)}\cdot n^{\cO(1)}\), where \(n=|P|\). Our algorithm relies on a structural analysis of how \(k\) rectangles interact with the point set, reducing \prbcshort\ to at most \(2^{\cO(k\log k)}\) instances of \ddmtcsp, each solvable in polynomial time. On the hardness side, we prove NP-completeness for boundary covering by axis-aligned \(L\)-shapes and use this reduction to establish NP-completeness of \prbcshort.
\end{abstract}

\newpage
\tableofcontents
\newpage
\section{Introduction}
Geometric covering problems ask for a small collection of geometric objects whose union ``covers'' a given set of points.  Geometric covering and packing problems have been extensively studied from the perspectives of algorithms and complexity.  Since many of these geometric set cover problems are NP-hard, they have also been investigated from the viewpoints of approximation algorithms and parameterized complexity~\cite{LangermanMorin,DBLP:journals/dam/HassinM91,DBLP:conf/esa/Marx05,DBLP:conf/iwpec/Marx06,DBLP:conf/iwpec/KaraK06}.  In this paper we focus on one such geometric set cover problem in the realm of parameterized complexity.

A particularly subtle variant is \emph{boundary covering}, where each input point must lie on the \emph{boundary} of at least one chosen object (rather than anywhere in its interior).  A key starting point for our work is a result of Langerman and Morin~\cite{LangermanMorin}, which is widely regarded as one of the seminal contributions to parameterized computational geometry.  In Section~5.2 of their paper, they study the following boundary-covering problem for spheres:

\medskip
\noindent{\textsc{Covering Points with Spheres}.}
\emph{Given a set $S$ of $n$ points in $\mathbb{R}^d$,
does there exist a set $\mathcal{H}$ of at most $k$ hyperspheres such that each point of $S$
lies on the surface of at least one hypersphere in $\mathcal{H}$?}

\smallskip

  They show fixed-parameter tractability (parameter $k$, for fixed $d$) by casting the problem as an instance of their framework \textsc{Dim-Set-Cover}.  The crucial geometric observation enabling this is a \emph{dimension-drop} property: if one considers the families $R_i$ of all $i$-spheres (points, pairs of points, circles, etc.), then the intersection of an $i$-sphere and a $j$-sphere (assuming neither contains the other) is an $\ell$-sphere for some $\ell < \min\{i,j\}$.  Thus, intersections strictly reduce dimension, which is exactly the structural condition that powers their algorithmic approach. In the planar case ($d=2$), this is a problem about covering points by \emph{circle boundaries}.  It is easy to misread this as a tractability statement for covering points by disks, but these problems are fundamentally different: requiring points to lie on the boundary is much more rigid than allowing them to lie anywhere in the interior.  In fact, the usual disk-cover variants remain W[1]-hard even under strong restrictions (e.g., unit disks; see, e.g.,~\cite{DBLP:conf/esa/Marx05}).  In particular, unless FPT = W[1], there is no algorithm running in time $f(k)\,n^{\Oh(1)}$ that, given $n$ points in the plane and an integer $k$, decides whether they can be covered by (at most) $k$ unit disks.

This distinction motivates us to systematically study boundary-covering questions for other geometric families. Langerman and Morin also point out that their dimension-drop methods do not apply to axis-parallel rectangles, since the intersection of two rectangles may again be a rectangle and therefore need not reduce dimension (see conclusion in~\cite{LangermanMorin}).   Motivated by this limitation, we investigate the boundary-covering analogue for axis-parallel rectangles, formalized as
follows.

\begin{tcolorbox}[enhanced,title={\color{black} {\bcdaprfull} (\bcdaprshort)}, colback=white, boxrule=0.5pt,
	attach boxed title to top center={xshift=.001cm, yshift*=-2.5mm},
	boxed title style={size=small,frame hidden,colback=white}]
	\textbf{Input:} A set $P$ of $n$ points, a set $\mathcal{R}$ of $m$ axis-parallel rectangles in the plane $\mathbb{R}^2$,
    and an integer $k\ge 0$.\\
    \textbf{Parameter:} $k$.\\
	\textbf{Question:} Does there exist a set of at most $k$ rectangles in $\mathcal{R}$ whose \emph{boundaries} cover all
    points in $P$?
\end{tcolorbox}

\paragraph*{Roadmap.}
We first study the parameterized complexity of the above \emph{discrete} rectangle variant and show that it is
\(\mathrm{W}[1]\)-hard with respect to the natural parameter \(k\). This motivates us to consider the \emph{continuous}
setting, where rectangles may be placed freely. For this variant we establish NP-completeness, and then show that,
despite this hardness, the problem admits a fixed-parameter algorithm parameterized by \(k\).

\paragraph*{Context and closest tractable analogue.}
While covering point sets by various geometric objects has been studied extensively, we are not aware of prior work on this boundary-covering variant for axis-parallel rectangles. The closest classical tractable special case is boundary covering by \emph{axis-parallel lines}: a minimum family of horizontal and vertical lines whose union covers all points can be found in polynomial time via a reduction to \textsc{Bipartite Vertex Cover} (equivalently, by computing a maximum matching and applying K\H{o}nig's theorem)~\cite{DBLP:journals/dam/HassinM91,gaur2007covering}. However, once we impose separate upper bounds on the number of lines parallel to the $x$-axis and to the $y$-axis, the problem becomes NP-complete via a simple reduction from \textsc{Constrained Bipartite Vertex Cover}(CBVC) (given a bipartite graph $G=(A,B,E)$ and integers $k_A,k_B$, the question is whether there exists a vertex cover using at most $k_A$ vertices from $A$ and at most $k_B$ vertices from $B$)~\cite{kuo2007efficient}.  Another closely related setting is where we are given a family of segments in the plane and asked to cover all points by this fixed set of segments.  This variant was recently studied by Kowalska and Pilipczuk~\cite{DBLP:conf/stacs/KowalskaP24}, who obtained a detailed parameterized complexity classification, showing that some versions admit FPT algorithms while others are hard.

\subsection{Our Results} 
Our first result establishes that \bcdaprfull\ (\bcdaprshort) is \(\mathrm{W}[1]\)-hard when parameterized by \(k\) even in $2$-dimension. 


\begin{theorem}
\bcdaprfull\ is \(\mathrm{W}[1]\)-hard when parameterized by \(k\).
\end{theorem}

Since the discrete rectangle variant is hard, we turn to the \emph{continuous} setting, where rectangles may be placed
freely; in particular, we study the following problem.

\begin{tcolorbox}[enhanced,title={\color{black} {\prbcfull} (\prbcshort)}, colback=white, boxrule=0.5pt,
	attach boxed title to top center={xshift=.10cm, yshift*=-2.5mm},
	boxed title style={size=small,frame hidden,colback=white}]
	\textbf{Input:} A finite set \(P=\{p_1,\dots,p_n\}\) of points in the plane \(\bR^2\), and a nonnegative integer
    \(k\in \mathbb{N}\cup\{0\}\).\\
    \textbf{Parameter:} \(k\).\\
	\textbf{Question:} Does there exist a family \(\mathcal{R}=\{R_1,\dots,R_{k'}\}\) of axis-parallel rectangles with
    \(k'\le k\) such that every point of \(P\) lies on the boundary of at least one rectangle in \(\mathcal{R}\)?
\end{tcolorbox}

We remark that Langerman and Morin~\cite{LangermanMorin} studied an analogous \emph{continuous} variant for spheres. Their paper explicitly presents an algorithm for the
continuous version.

As noted above, boundary covering by \emph{axis-parallel lines} (when lines may be chosen freely) is solvable in
polynomial time, whereas the corresponding \emph{discrete} version (where one must choose from a given set of
axis-parallel lines) is NP-complete. Our next contribution shows that, even in the continuous setting, moving beyond
lines to slightly richer axis-parallel shapes already leads to intractability. We first prove NP-completeness for
boundary covering by axis-aligned L-shapes, \prbcLfull\ (\prbcLshort), where an L-shape is the union of one
horizontal and one vertical segment sharing an endpoint. We then use this as a starting point to obtain NP-completeness
for \prbcshort.

\begin{theorem}
\label{thm:lshape-npc}
\prbcLshort is NP-complete.
\end{theorem}

\begin{theorem}
\label{thm:prbc-npc}
\prbcshort\ is NP-complete.
\end{theorem}

\textsc{Constrained Bipartite Vertex Cover} is our starting point for the reduction to \prbcLshort~\cite{kuo2007efficient}. 
Theorem~\ref{thm:lshape-npc} provides the core hardness gadget: it captures the essential ``turning'' behavior that
rectangle boundaries must simulate.  Building on this gadget, our reduction for Theorem~\ref{thm:prbc-npc} encodes each
L-shape choice using a constant number of rectangle-boundary constraints while preserving the parameter \(k\).

Having established NP-completeness, we study \prbcshort\ from the perspective of parameterized complexity with respect
to the solution size \(k\). Our main algorithmic result shows that, despite NP-hardness, the problem is fixed-parameter
tractable parameterized by \(k\).

\begin{theorem}
\label{thm:prbc-fpt}
\prbcshort\ is fixed-parameter tractable when parameterized by \(k\) and admits an algorithm with running time
\(2^{\Oh(k \log k)} n^{\Oh(1)}\), where \(n\) is the number of input points.
\end{theorem}

This result is obtained by carefully analyzing how the $k$ rectangles in a solution can interact with the point set. We first discretize the plane and compute a set of at most $4k$ axis-parallel lines that together contain all input points. If size of the set exceeds $4k$, we directly return \no. We then guess which side of which rectangle aligns with which chosen line, and how the points on each such line are covered by the horizontal sides of the rectangles aligned with it. Points that are not covered by these aligned rectangles are called \emph{exceptional points}; we show that their total number is bounded by $2k$, and we guess how these exceptional points are covered by sides of rectangles that do not align with the chosen lines. 

These guesses allow us to encode the problem as an instance of \ddmtcsp, which is known to be solvable in polynomial time~\cite{DBLP:journals/dcg/AgrawalKLSZ24}. In summary, we obtain at most $2^{\Oh(k \log k)}$ instances of \ddmtcsp  such that the input instance is a \yes-instance if and only if at least one of these \ddmtcsp  instances is satisfiable. This also illustrates the power of \ddmtcsp as a tool for designing FPT algorithms in computational geometry~\cite{DBLP:journals/dcg/AgrawalKLSZ24}. We give a more detailed technical overview of this result in Section~\ref{sec:overview}.

\subsection{Related work.}
Many classical geometric set cover problems remain hard when parameterized by the number $k$ of objects: for example, covering points with unit squares is W[1]-hard (as shown by Marx, see also later expositions)~\cite{DBLP:conf/esa/Marx05,DBLP:conf/stacs/KowalskaP24}, and related separator-based techniques give $n^{\Oh(\sqrt{k})}$-time algorithms for covering with disks/squares rather than FPT running times~\cite{DBLP:journals/talg/MarxP22}.  Covering points by lines also has a rich literature (the general \textsc{Point Line Cover} problem is NP-hard and has been studied from a parameterized/kernelization perspective)~\cite{DBLP:journals/dam/HassinM91,LangermanMorin,DBLP:journals/talg/KratschPR16}.  
\section{Technical Overview of Our Results}
\label{sec:overview}
In this section we provide a high-level overview of our main results and the techniques behind them. We begin with the
\(\mathrm{W}[1]\)-hardness proof for the discrete rectangle variant \bcdaprshort, then outline the NP-hardness proofs for
the continuous setting (via axis-aligned L-shapes and then rectangles), and finally summarize the ideas underlying
our fixed-parameter algorithm for \prbcshort.

\subsection{Technical overview for W[1]-hardness for \bcdaprshort}
We prove \(\mathrm{W}[1]\)-hardness of \bcdaprshort\ by a parameterized reduction from \rcsp, where each constraint
involves exactly two variables and each variable appears in exactly three constraints (see~\cref{sec:constraintSatisfaction} for formal definition). 
The input to this problem is a set $Z = \{z_1$, $z_2$, $\cdots$, $z_{k'}\}$ of variables, a common (finite) domain $D \subseteq \mathbb{N}$ for all the variables $z_i\in Z$, a set $\mathcal{C}$ of $m'=\frac{3k'}{2}$ constraints and for each constraint $C_{ij}\in\mathcal{C}$, involving variables $z_i$ and $z_j$, a set ${\sf sat}_{ij}\subseteq D\times D$ of pairs of assignments to $z_i$ and $z_j$ that satisfy $C_{ij}$. The objective is to check if there is an assignment ${\sf asg}: Z \rightarrow D$ such that each constraint in $\mathcal{C}$ is satisfied.
Let
\(I=(Z,D,\mathcal{C}, \{{\sf sat}_c\}_{c\in\mathcal{C}})\) be an \rcsp instance.
We construct a set of points \(P\), a family of axis-parallel rectangles \(\mathcal{R}\), and a budget
\(k := 4k' + m' = \tfrac{11k'}{2}\), such that \(I\) is satisfiable if and only if the constructed instance admits a
boundary cover by at most \(k\) rectangles.

\smallskip

\noindent\textbf{Variable gadgets.}
For each variable \(z_i\) we place a square \({\sf SQ}_i\) of side length \(t+2\).
Along each side of \({\sf SQ}_i\) we place \(t\) \emph{variable points}, one for each domain value, in cyclic order
around the square (plus one dummy point to separate choices). For every \(a_j\in D\) we introduce four
\emph{variable rectangles} \(BL^{(i)}_j, LT^{(i)}_j, TR^{(i)}_j, RB^{(i)}_j\), each defined by a pair of diagonally
opposite variable points on two consecutive sides of \({\sf SQ}_i\).
These four rectangles are pairwise disjoint and together cover all variable points of \({\sf SQ}_i\).
Crucially, covering all variable points forces choosing \emph{exactly one} such quadruple:
with fewer than four variable rectangles some variable points remain uncovered, and with four rectangles the only way to
cover all variable points is to pick one rectangle of each type with a \emph{common index} \(j\), i.e.,
\(\{BL^{(i)}_j,LT^{(i)}_j,TR^{(i)}_j,RB^{(i)}_j\}\).
Thus, selecting these four rectangles encodes the assignment \(z_i=a_j\). See Figure~\ref{var_gad} for an illustration.

\smallskip
\noindent\textbf{Constraint gadgets.}
Suppose \(z_i\) appears in the three constraints \(C_{ij_1},C_{ij_2},C_{ij_3}\).
For each domain value \(a_p\in D\), we place three \emph{constraint points}
\(\beta^{(ij_1)}_p,\beta^{(ij_2)}_p,\beta^{(ij_3)}_p\) on the bottom side of \({\sf SQ}_i\), between the consecutive
variable points corresponding to \(a_p\) and \(a_{p+1}\). By design, when the variable gadget selects the quadruple
encoding \(z_i=a_p\), these three constraint points are \emph{not} covered by the chosen variable rectangles and must be
covered by additional rectangles, see Figure~\ref{fig:cons_pts_mod} for an illustration.
For each constraint \(C_{ij}\) we introduce a dedicated horizontal line segment \(h_{ij}\) above the variable gadgets
and place on it a \emph{guard point} \(\alpha_{ij}\) that cannot be covered vertically (refer to \Cref{cons_rec}). For every satisfying pair
\((a_p,a_q)\) of \(C_{ij}\), we add a \emph{constraint rectangle} whose boundary covers \(\alpha_{ij}\) horizontally and
covers \(\beta^{(ij)}_p\) and \(\beta^{(ji)}_q\) vertically. Because each \(\alpha_{ij}\) lies on its own line \(h_{ij}\),
covering all guard points forces the solution to pick at least one rectangle per constraint.
See Figure~\ref{fig:final_rec} for illustration.

\smallskip
\noindent\textbf{Budget forcing and correctness.}
We set \(k=4k'+m'\). Any feasible solution must cover all variable points, implying at least \(4\) variable rectangles
per gadget and hence at least \(4k'\) variable rectangles overall. Likewise, each guard point \(\alpha_{ij}\) can only be
covered by a constraint rectangle associated with \(C_{ij}\), so at least \(m'\) constraint rectangles are necessary.
Since the total budget is exactly \(4k'+m'\), every size-\(k\) solution must be \emph{tight}: it selects exactly one
quadruple of variable rectangles per variable and exactly one constraint rectangle per constraint.
This yields a bijection between solutions and assignments: from a satisfying assignment \(\sigma\) we select the
corresponding \(4k'\) variable rectangles and, for each constraint, the unique rectangle corresponding to the satisfying
pair \((\sigma_i,\sigma_j)\); conversely, from any size-\(k\) solution we read off a unique value \(\pi(i)\) per variable
gadget and argue that for every constraint \(C_{ij}\) the selected constraint rectangle exists only if
\((a_{\pi(i)},a_{\pi(j)})\) satisfies \(C_{ij}\). Hence the constructed \bcdaprshort\ instance is a \yes instance if and
only if the original \rcsp instance is satisfiable, establishing \(\mathrm{W}[1]\)-hardness parameterized by \(k\). This result is given in Section \ref{sec:woneHardness}. 

\subsection{Overview of NP-hardness results}
Our NP-hardness proof proceeds in two stages. We begin by formally defining an axis-aligned L-shape.

\begin{definition}[Axis-aligned L-shape]\label{def:lshape}
An \emph{axis-aligned L-shape} is a set \(L = V \cup H\), where \(V\) is a vertical line segment and \(H\) is a
horizontal line segment, such that \(V \cap H=\{c\}\) for some point \(c\), and \(c\) is an endpoint of both \(V\) and
\(H\) (the \emph{corner}). We call $H$ the \emph{horizontal arm} of $L$ and $V$ the \emph{vertical arm} of $L$. A point \(p\in\mathbb{R}^2\) is \emph{covered} by L if \(p\in L\).
\end{definition}

We now define the associated covering problem.

\begin{tcolorbox}[enhanced,title={\color{black} {\plscfull} (\plscshort)}, colback=white, boxrule=0.5pt,
	attach boxed title to top center={xshift=.001cm, yshift*=-2.5mm},
	boxed title style={size=small,frame hidden,colback=white}]
	\textbf{Input:} A finite set \(P\subseteq \bR^2\) of \(n\) points and an integer \(k\ge 0\).\\
    \textbf{Parameter:} \(k\).\\
	\textbf{Question:} Do there exist at most \(k\) axis-aligned L-shapes whose union covers all points of \(P\)?
\end{tcolorbox}

We first show that \plscshort\ is NP-complete via a reduction from \pcbvcfull.
Given a constrained bipartite vertex cover instance $(G=(A,B,E),k_{A},k_{B})$, we embed the edges as points on an $[m]\times[n]$ grid and add two families of \emph{guard points}: $k_{B}$ long vertical columns to the top-left and $k_{A}$ long horizontal rows to the bottom-right.  Choosing an L-shape with a vertical arm on a designated column encodes selecting a vertex in $B$, and choosing an L-shape with a horizontal arm on a designated row encodes selecting a vertex in $A$.  By setting $M=k+1$ guard points per designated line, any solution with at most $k=k_{A}+k_{B}$, L-shapes is forced to dedicate $k_{B}$ distinct shapes to the top-left columns and $k_{A}$ distinct shapes to the bottom-right rows.  The remaining arms in the central grid then correspond exactly to a constrained vertex cover of $G$. For a schematic geometric overview of our reduction see Figure~\ref{fig:lsc-overview}. 

In the second stage, we reduce from \pnelscfull(\pnelscshort), the restriction of \plscshort to $\NE$-oriented L-shapes, to \prbcfull.  Starting from a point set $P_0$ and parameter $k$, we place $k$ vertical guard columns strictly to the right and $k$ horizontal guard rows strictly above $P_0$, each populated with $M=2k+1$ guard points.  Any family of at most $k$ rectangles that boundary-covers all guards must place a rectangle side on each guard line (see Figure~\ref{fig:brc-guards} for an illustration), and a simple geometric argument shows that every such rectangle necessarily uses one vertical guard line as its right side and one horizontal guard line as its top side.  Consequently, each rectangle “pays for’’ exactly two guard sets and its left and bottom sides form an $\NE$-oriented L-shape that lies over $P_0$.  This yields a one-to-one correspondence between rectangle solutions for $P_0$ and $\NE$-L-shape solutions for $P_0$, completing the NP-completeness proof for \prbcshort. These results are proved in Section~\ref{sec:npHardness}.

\subsection{Overview of an FPT algorithm for \prbcshort}
We give an FPT algorithm for \prbcshort\ parameterized by \(k\). The algorithm can be viewed as a parameterized
reduction to \ddmtcsp: given an input instance \((P,k)\), we compute, in time \(f(k)\cdot |P|^{\Oh(1)}\), an equivalent
instance of \ddmtcsp. Consequently, \prbcshort\ is fixed-parameter tractable with respect to \(k\). Before proceeding, we
first define \ddmtcsp.

\paragraph*{{Distinct Domain Monotone $2$-CSP}:} The input to this problem is a set $Z = \{z_1$, $z_2$, $\cdots$, $z_{n'}\}$ of variables, a (finite) domain $D_i \subseteq \mathbb{N}$ for the variable $z_i$, for each $i\in [n']$, and a set $C$ of $m'$ constraints, where each constraint is of the following form: $c = [z_i \si f(z_j)]$, where $i,j \in [n']$, $f: D_{j} \rightarrow \mathbb{N}$ is a monotone function, and $\si \in \{\leq, \geq, =\}$. The objective is to check if there is an assignment ${\sf asg}: Z \rightarrow \mathbb{N}$ such that: i) for each $i \in [n']$, ${\sf asg}(z_i) \in D_i$, and ii) each constraint in $C$ is satisfied, i.e., for each $c = [z_i \si f(z_j)] \in C$, ${\sf asg}(z_i) \si f({\sf asg}(z_j))$ is true.

In the above problem definition, for simplicity in its usage, we allow domains to be empty sets, in which case we trivially have a no-instance of the problem. The problem \mtcsp is a special case of {\sc Distinct Domain Monotone $2$-CSP}, where the domains of all the variables are the same. A polynomial time algorithm for {\sc Monotone $2$-CSP} can be obtained via a simple reduction to {\sc $2$-SAT}~\cite{DBLP:journals/dcg/AgrawalKLSZ24}. A very minor modification to this algorithm for {\sc Monotone $2$-CSP} results in a polynomial time algorithm for {\sc Distinct Domain Monotone $2$-CSP}, which is stated in the following proposition.

\begin{proposition}[\cite{DBLP:journals/dcg/AgrawalKLSZ24}]\label{prop:distinct-domian-csp}
{\sc Distinct Domain Monotone $2$-CSP} has a polynomial time algorithm.
\end{proposition}

\noindent
Having defined \ddmtcsp, we are now ready to describe the different steps in our reduction to \ddmtcsp.
\begin{itemize}
\setlength{\itemsep}{3pt}
\item \textbf{Step 1: Discretization and a finite candidate family.}
We discretize the plane by selecting $2|P|$ vertical and $2|P|$ horizontal lines, and show that there exists an optimal solution
in which \emph{every} rectangle side lies on one of these chosen lines.
As a consequence, the universe of candidate rectangles is finite; in particular, there is a family of at most $|P|^{\mathcal{O}(1)}$
rectangles (e.g., $\mathcal{O}(|P|^4)$) such that some optimal solution of size at most $k$ can be chosen entirely from this family.

\item \textbf{Step 2: A small set of ``important'' lines ($\impline$.}
Next we show that if $(P,k)$ is a \yes-instance of \prbcshort, then there exists a set of at most $4k$ lines (vertical and horizontal)
whose union contains all points of $P$.
Moreover, such a set can be found using a known routine via a reduction to \textsc{Vertex Cover} on a bipartite graph.
We denote the resulting set of lines by $\impline$.

\item \textbf{Step 3: Exceptional points on a line.}
Fix a line $L\in\impline$. Observe that if a solution rectangle has \emph{no} side aligned with $L$, then its boundary intersects $L$
in at most two points. Therefore, among $k$ rectangles, at most $2k$ points of $P\cap L$ can be covered by rectangles that do not align
with $L$.
We call a point of $P\cap L$ \emph{exceptional} (with respect to a solution) if it is covered by a rectangle whose relevant side is
\emph{not} aligned with $L$. Hence, every $L\in\impline$ has at most $2k$ exceptional points.
In particular, if $L$ contains more than $2k$ points of $P$, then some rectangle side in any feasible solution must align with $L$.

\item \textbf{Step 4: Guessing a skeleton.}
Let $R_1,\dots,R_{k'}$ be a hypothetical solution with $k'\le k$.
We define a \emph{skeleton} to be a function that maps each side of each $R_i$ to an element of $\impline\cup\{\bot\}$:
\begin{itemize}
\item mapping a side to $L\in\impline$ means that we guess this side lies on $L$;
\item mapping a side to $\bot$ means that the side does not align with any line in $\impline$.
\end{itemize}
Intuitively, the skeleton guesses which rectangle sides are ``aligned'' to important lines and which sides are ``free''.

\item \textbf{Step 5: Ordering endpoints on each important line.}
Fix a horizontal line $h\in\impline$.
From the skeleton we know exactly which rectangle sides (top or bottom sides) lie on $h$.
Suppose there are $m$ such sides. Each contributes two endpoints; denote them by $l_i$ and $r_i$ for $i\in[m]$.
We guess the left-to-right order of these $2m$ endpoints along $h$, i.e., a permutation of $\{l_i,r_i:i\in[m]\}$.

We declare such an order \emph{admissible} if, after replacing each $l_i$ by \textsf{left} and each $r_i$ by \textsf{right},
the resulting string belongs to the Dyck language over $\Sigma=\{\textsf{left},\textsf{right}\}$.

\item \textbf{Step 6: Gaps and counting exceptional points.}
For each admissible Dyck string, consider its unique decomposition into \emph{minimal} Dyck blocks,
say $X = X_1X_2\cdots X_q$, where each $X_j$ is a nonempty Dyck word and no proper prefix of $X_j$ is a Dyck word.
This induces $q+1$ \emph{gaps}: before $X_1$, between consecutive blocks, and after $X_q$.
We then guess, for each gap, \emph{how many} exceptional points of $P$ lie in that gap (but not their identities).
Since the number of exceptional points on $h$ is at most $2k$, we only need to guess how they are distributed among the $q{+}1$ gaps.
Formally, we guess an integer $k^\star\le 2k$ and a $(q{+}1)$-tuple of nonnegative integers
\[
(w_1,\ldots,w_{q+1})\in \mathbb{Z}_{\ge 0}^{\,q+1}
\quad\text{such that}\quad
\sum_{i=1}^{q+1} w_i \;=\; k^\star.
\]
The number of such distributions is $\binom{k^\star+q}{q}$, and since $k^\star \le 2k$ and $q \le m \le k$, this quantity is bounded by $8^k$. A tighter analysis gives an improved bound of $6.75^k$. 

\smallskip
\item \textbf{Step 7: Assigning exceptional points to vertical sides.}
Finally, for each exceptional point (abstractly counted in Step~6), we guess which rectangle covers it and whether it is covered
by the left or the right vertical side of that rectangle.

\item \textbf{Step 8: Reduction to \ddmtcsp.}
Once the skeleton, endpoint orders, gap counts, and exceptional-point assignments are fixed, the remaining feasibility constraints
can be expressed as an instance of \ddmtcsp.  We solve that \ddmtcsp instance and accept if any branch succeeds. This reduction step is fairly technical, so we omit its details from the overview.

\end{itemize}

\medskip
\noindent
The number of branches created by the guesses above is bounded by $2^{\Oh(k \log k)}\cdot n^{\mathcal{O}(1)}$, and each branch can be processed
in polynomial time in $n$ (plus the time needed to solve the resulting \ddmtcsp instance).
Therefore the overall algorithm runs in  $2^{\Oh(k \log k)}\cdot n^{\mathcal{O}(1)}$ time.
\section{Notations and Preliminaries}
For $m \in \mathbb{N}$, we denote $\{1, \dots, m\}$ by $[m]$ and $\{0, \dots, m\}$ by $[m]_0$. 
An axis-parallel rectangle $R$ in the plane has the following form $$R=\{(x,y)\in\bR^2\;|\;a\leq x\leq b,\; c\leq y\leq d\},$$ for some $a,b,c,d\in\bR$. Throughout the paper, all rectangles are assumed to be axis-parallel, so we simply refer to them as rectangles. The boundary of a rectangle $R$, denoted by $\bd(R)$, is defined as $$\bd(R)=\{(x,y)\in R\;|\; x=a \text{ or } x=b \text{ or } y=c \text{ or } y=d\}.$$ 

\noindent
The four lines, namely, $x=a$, $x=b$, $y=c$ and $y=d$ are termed as the \emph{boundary lines} of $R$. Furthermore, we call the 
sides of a rectangle $R$ as \emph{left, right, bottom} and \emph{top side} of $R$ and the corresponding lines on which they lie, we call them the \emph{left, right, bottom} and \emph{top boundary line}  of $R$ respectively. A point $p \in P$ is said to be \emph{covered} by a rectangle $R$ if $p$ lies on the boundary of $R$.

Let $\pi$ be an ordering of the elements in the set $\{e_1, \dots, e_n\}$. 
We write $x \prec_{\pi} y$ to indicate that element $x$ appears to the left of element $y$ in the order induced by $\pi$. 
Accordingly, we represent the ordering $\pi$ as
$$
e_1 \prec_{\pi} e_2 \prec_{\pi} \dots \prec_{\pi} e_n.
$$
When the ordering is clear from context, we omit the subscript and simply write $\prec$.

\subsection{Definition of {\sc $3$-Regular $2$-CSP}}\label{sec:constraintSatisfaction}
We now define the source problem for our W[1]-hardness.  
Toward that we first define constraint satisfaction problems (CSPs) of arity two (also called binary constraints).) We follow the notation and definitions of
the seminal paper of Guruswami et al.~\cite{DBLP:conf/stoc/GuruswamiLRS024}. 
Formally, a CSP instance $G$ is a quadruple $(V(G), E(G), D, \{C_{e}\}_{e\in E(G)})$, where:
\begin{itemize}
\setlength{\itemsep}{2pt}
    \item $V(G)$ is the set of variables.
    \item $E(G)$ is the set of constraints.
    Each constraint $e=\cbra{u_e,v_e} \in E(G)$ has arity $2$ and is related to two distinct variables $u_e,v_e\in V(G)$.
    The \emph{constraint graph} is the undirected graph on the vertices $V(G)$ and the edges $E(G)$. Note that we allow multiple constraints between the same pair of variables and thus the constraint graph may have parallel edges.
    \item $D$ is for the alphabet of each variable in $V(G)$. We use $D=[n]$. 
    \item  Given a constraint $e\in E(G)$, $C_e\subseteq D \times D =[n]\times [n]$. Furthermore, given $\{C_{e}\}_{e\in E(G)}$, we can define $\{\Pi_{e}\}_{e\in E(G)}$, the set of validity functions of constraints.  
    Given a constraint $e\in E(G)$, the validity function $\Pi_e(\cdot,\cdot)\colon D\times D\to \bin$ checks whether the constraint $e$ between $u_e$ and $v_e$ is satisfied. That is, $\Pi_e(\cdot,\cdot)$ assigns $1$ if and only of the  tuple is in $C_e$. 
\end{itemize}
We use $|G|=(|V(G)|+|E(G)|)\cdot|D|$ to denote the \emph{size} of a CSP instance $G$.

\paragraph*{Assignment and Satisfaction Value.}

An \emph{assignment} is a function $\sigma\colon V(G)\to D$ that assigns to each variable a value from the alphabet. The \emph{satisfaction value} for a mapping $\sigma$, denoted as $\val(G,\sigma)$, represents the proportion of constraints that $\sigma$ satisfies, i.e., \[\val(G, \sigma)=\frac{1}{|E|}\sum_{e\in E} \Pi_e(\sigma(u_e),\sigma(v_e)).\] 
The satisfaction value for $G$, denoted by $\val(G)$, is the highest satisfaction value across all mappings, i.e., $\val(G) = \max_{\sigma\colon V(G)\to \Sigma} \val(G, \sigma)$. We define an assignment $\sigma$ as a \emph{solution} to a CSP instance $G$ if $\val(G,\sigma) = 1$, and say $G$ is \emph{satisfiable} if and only if $G$ has a solution. When the context is clear, $\sigma$ is omitted in the constraint description; therefore, $\Pi_e(u_e,v_e)$ represents $\Pi(\sigma(u_e),\sigma(v_e))$.

We will be working with the {\sc $3$-Regular $2$-CSP} problem. An input to this problem consists of a 2CSP $G$ with $k$ variables over size-$n$ alphabets, where the constraint graph $G$ is $3$-regular. The question is whether $G$ is satisfiable. 

\begin{proposition}[\cite{DBLP:conf/stoc/GuruswamiLRS024,DBLP:conf/soda/LokshtanovR0Z20}]\label{prop:distinct-domain-2-csp}
\rcsp is \(\mathrm{W}[1]\)-hard when parameterized by the number of variables.
\end{proposition}

In short, we consider \rcsp, where every variable appears in exactly three constraints and every constraint involves
exactly two variables.

\subsection{Definition of Distinct Domain Monotone $2$-CSP}
\label{sec:DDMCSP}
The input to this problem is a set $Z = \{z_1$, $z_2$, $\cdots$, $z_{n'}\}$ of variables, a (finite) domain $D_i \subseteq \mathbb{N}$ for the variable $z_i$, for each $i\in [n']$, and a set $C$ of $m'$ constraints, where each constraint is of the following form: $c = [z_i \si f(z_j)]$, where $i,j \in [n']$, $f: D_{j} \rightarrow \mathbb{N}$ is a monotone function, and $\si \in \{\leq, \geq, =\}$. The objective is to check if there is an assignment ${\sf asg}: Z \rightarrow \mathbb{N}$ such that: i) for each $i \in [n']$, ${\sf asg}(z_i) \in D_i$, and ii) each constraint in $C$ is satisfied, i.e., for each $c = [z_i \si f(z_j)] \in C$, ${\sf asg}(z_i) \si f({\sf asg}(z_j))$ is true.

In the above problem definition, for simplicity in its usage, we allow domains to be empty sets, in which case we trivially have a no-instance of the problem. The problem \mtcsp is a special case of {\sc Distinct Domain Monotone $2$-CSP}, where the domains of all the variables are the same. A polynomial time algorithm for {\sc Monotone $2$-CSP} can be obtained via a simple reduction to {\sc $2$-SAT}~\cite{DBLP:journals/dcg/AgrawalKLSZ24}. A very minor modification to this algorithm for {\sc Monotone $2$-CSP} results in a polynomial time algorithm for {\sc Distinct Domain Monotone $2$-CSP}, which is stated in the following proposition.

\begin{proposition}[\cite{DBLP:journals/dcg/AgrawalKLSZ24}]\label{prop:distinct-domian-csp}
{\sc Distinct Domain Monotone $2$-CSP} has a polynomial time algorithm.
\end{proposition}








	

\section{W[1]-hardness for \bcdaprshort}
\label{sec:woneHardness}
In this section, we prove that the problem becomes \(\mathrm{W}[1]\)-hard when the rectangles are part of the input.
That is, we show that \bcdaprfull\ (\bcdaprshort) is \(\mathrm{W}[1]\)-hard parameterized by the solution size \(k\).
Our proof is via a parameterized reduction from \textsc{3-Regular 2-Constraint Satisfaction Problem} (\rcsp).

\subsection{Reduction from {\sc $3$-Regular $2$-CSP}}
The reduction is from \rcsp instance, where  each constraint consists of exactly $2$ variables and each variable appears in exactly $3$ constraints.

Let $I$ be an instance to $3$-regular $2$-CSP with the set of variables $Z = \{z_1, \dots, z_{k'}\}$ where each variable $z_i$ takes value from a domain $D = \left\{a_1, \dots, a_t\right\}$, a set of constraints $\mathcal{C}$ with $|\mathcal{C}|=m'$, where each constraint contains $2$ variables and each variable appears in exactly $3$ constraints (so $m'=\frac{3k'}{2}$) and for each constraint $C_{ij}\in\mathcal{C}$ involving variables $z_i$ and $z_j$ a set ${\sf sat}_{ij}\subseteq D\times D$ of pairs of assignments to $z_i$ and $z_j$ that satisfy $C_{ij}$.. We construct an instance  $(P, \mathcal{R},k)$ of \bcdaprshort as follows.

\paragraph*{\underline{\large Variable Gadgets:}}

For each variable $z_i$, we construct the variable gadget as follows.

\begin{figure}[htbp]
    \centering
\includegraphics[width=0.5\linewidth]{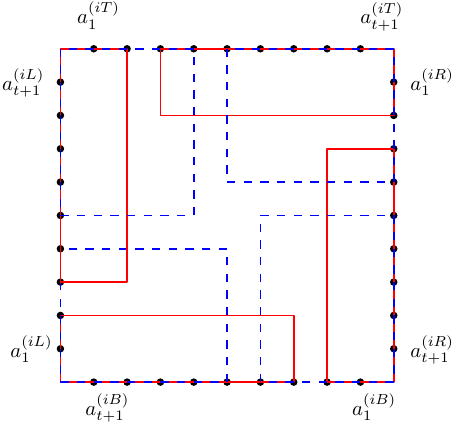}
    \caption{Variable gadget ${\sf SQ}_i$ corresponding to the variable $z_i$}
\label{var_gad}
\end{figure}

\begin{itemize}
    \item The variable gadget corresponding to $z_i$ is a square ${\sf  SQ}_i$ of side length $t+2$ embedded in the plane $\mathbb{R}^2$. 
We denote its four sides by $L$, $R$, $T$, and $B$, representing the \emph{left}, 
\emph{right}, \emph{top}, and \emph{bottom} sides, respectively. We describe the construction for one side of the square ${\sf SQ}_i$, say $L$; the
others are defined analogously. Corresponding to the domain
$D = \left\{a_1, \dots, a_{t}\right\}$,
we introduce a set of points $\left\{a^{(iL)}_1, \dots, a^{(iL)}_{t}\right\}$,
one point for each element of the domain, preserving the natural order of the
domain elements. We add one additional dummy point $a^{(iL)}_{t+1}$ on top of of $a^{(iL)}_{t}$. These points are placed on side $L$ such that any two consecutive points are at unit distance from each other.

We repeat this placement for each of the four sides of the square, thereby creating four copies of every domain element, one on each side. For any two consecutive sides of the square, the placement of these copies follows a
\emph{cyclic order} (see Figure~\ref{var_gad}). We refer to all these points collectively as the \emph{variable points}.

\item Corresponding to each domain value $a_j \in D$, we create four types of
\emph{variable rectangles} denoted by
$BL^{(i)}_j$, $LT^{(i)}_j$, $TR^{(i)}_j$, and $RB^{(i)}_j$. Each rectangle is uniquely
determined by a pair of diagonally opposite points lying on two corresponding sides of
${\sf SQ}_i$. 

For example, the rectangle $BL^{(i)}_j$ is defined by the two points
$a^{(iB)}_{j+1}$ and $a^{(iL)}_j$ lying on the bottom and left sides of ${\sf SQ}_i$,
respectively. These two points uniquely determine the axis-parallel rectangle
$BL^{(i)}_j$. Similarly, we define the other three rectangles as follows:
\begin{itemize}
    \item $LT^{(i)}_j$ is defined by the two points
    $a^{(iL)}_{j+1}$ and $a^{(iT)}_j$;
    \item $TR^{(i)}_j$ is defined by the two points
    $a^{(iT)}_{j+1}$ and $a^{(iR)}_j$;
    \item $RB^{(i)}_j$ is defined by the two points
    $a^{(iR)}_{j+1}$ and $a^{(iB)}_j$.
\end{itemize}

\item Observe that for each variable $z_i$ and each domain value $a_j \in D$, the four
rectangles
\[
\{BL^{(i)}_j,\, LT^{(i)}_j,\, TR^{(i)}_j,\, RB^{(i)}_j\}
\]
are pairwise disjoint and together cover all variable points placed on the boundary of
the square ${\sf SQ}_i$. Selecting these four rectangles corresponding to a domain value
$a_j$ represents assigning the variable $z_i$ the value $a_j$. For example, in
Figure~\ref{var_gad}, selecting the four red rectangles corresponds to the assignment
$z_i = a^{(i)}_2$, while selecting the four dashed blue rectangles corresponds to the
assignment $z_i = a^{(i)}_{4}$.

\item Furthermore, for a fixed variable $z_i$, unless all four rectangles corresponding to
some domain value $a_j \in D$ are selected, all the variable points on the corresponding
square ${\sf SQ}_i$ are not covered. Hence, covering all variable points of ${\sf SQ}_i$
forces the selection of all four rectangles associated with exactly one domain value.
\end{itemize}

We arrange the variable gadgets, namely the squares
${\sf SQ}_1, \dots, {\sf SQ}_{k'}$, next to each other so that their
bottom sides lie on the same horizontal line; see
Figure~\ref{arr_var_gad}.

\begin{figure}[htbp]
    \centering
    \includegraphics[width=0.5\linewidth]{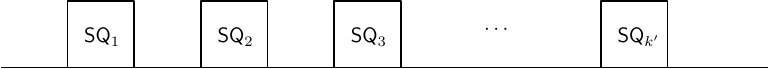}
    \caption{Arrangement of variable gadgets}
    \label{arr_var_gad}
\end{figure}

\paragraph*{\underline{\large Constraint Gadgets:}}  
For each constraint $C_{ij}\in\mathcal{C}$ containing the two variables $z_i$ and $z_j$, we add some additional points and rectangles namely, \emph{constraint points} and \emph{constraint rectangles} respectively.

\begin{itemize}
    \item Consider a variable gadget ${\sf SQ}_i$ corresponding to the variable $z_i$. Suppose $z_i$ appears in the constraints $C_{ij_1}, C_{ij_2}$ and $C_{ij_3}$. Now for each domain value $a_p\in D$, on bottom side of ${\sf SQ}_i$ we add $3$ corresponding constraint points $\beta^{(ij_1)}_p$, $\beta^{(ij_2)}_p$, $\beta^{(ij_3)}_p$ between the variable points $a^{(iB)}_p$ and $a^{(iB)}_{p+1}$.

Observe that, there are exactly $3$ constraint points in ${\sf SQ}_i$ between $a^{(iB)}_j$ and $a^{(iB)}_{j+1}$ 
 that are not covered by any of the four variable rectangles
corresponding to the assignment $z_i = a_p$. For reference, see \Cref{fig:cons_pts_mod}.
\begin{figure}[htbp]
    \centering \includegraphics[width=0.4\linewidth]{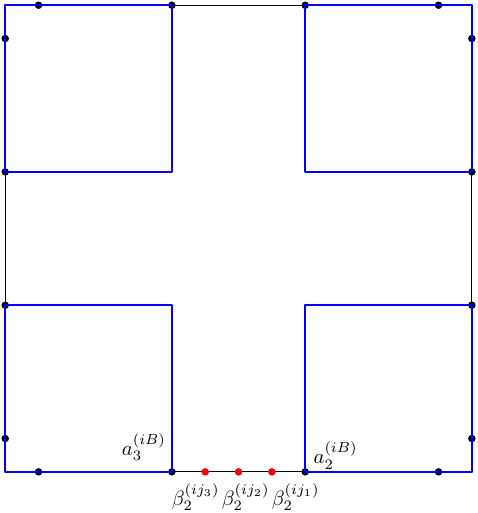}
    \caption{Constraint points corresponding to domain value $a_2$ in the variable gadget ${\sf SQ}_i$}
\label{fig:cons_pts_mod}
\end{figure}

\begin{figure}[htbp]
    \centering
\includegraphics[width=0.5\linewidth]{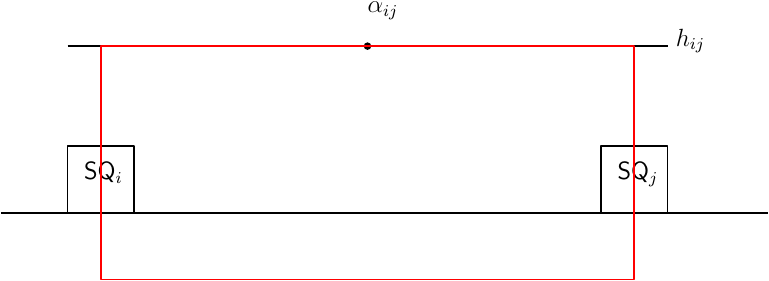}
    \caption{Constraint guard point $\alpha_{ij}$ corresponding to variables $z_i, z_j$}
    \label{cons_rec}
\end{figure}

    \item For every constraint $C_{ij} \in \mathcal{C}$, we add a unique horizontal line segment $h_{ij}$ above the variable gadgets, spanning from ${\sf SQ}_i$ to ${\sf SQ}_j$.  
    We make sure no constraint rectangles except those corresponding to the constraint $C_{ij}$ will have a horizontal side on this line. 
    We place a \emph{constraint guard point} $\alpha_{ij}$ on $h_{ij}$ such that no rectangle covers it vertically, see figure~\ref{cons_rec}.

    \item We now construct \emph{constraint rectangles} to cover the constraint points
together with the corresponding guard points. For every pair of assignments
$(a_p, a_q) \in D \times D$ that satisfies the constraint
$C_{ij}$, we add a constraint rectangle defined as follows. The rectangle is chosen so that it covers the point $\alpha_{ij}$ horizontally,
and the two points $\beta^{(ij)}_p$ and $\beta^{(ji)}_q$ vertically. In this
way, the rectangle simultaneously covers the constraint guard point and the two associated constraint points in the variable gadgets.

For example, in Figure~\ref{fig:final_rec}, if the assignments
$z_i = a_{1}$ and $z_j = a_{2}$ satisfy the constraint $C_{ij}$,
we add the green (dash-dotted) constraint rectangle that covers the red constraint point $\beta^{(ij)}_{1}$ corresponding to $C_{ij}$ in
${\sf SQ}_i$ and the constraint blue point $\beta^{(ji)}_{2}$ corresponding to $C_{ij}$ in ${\sf SQ}_j$ vertically,
and the point $\alpha_{ij}$ horizontally.

\begin{figure}[htbp]
    \centering
    \includegraphics[width=0.8\linewidth]{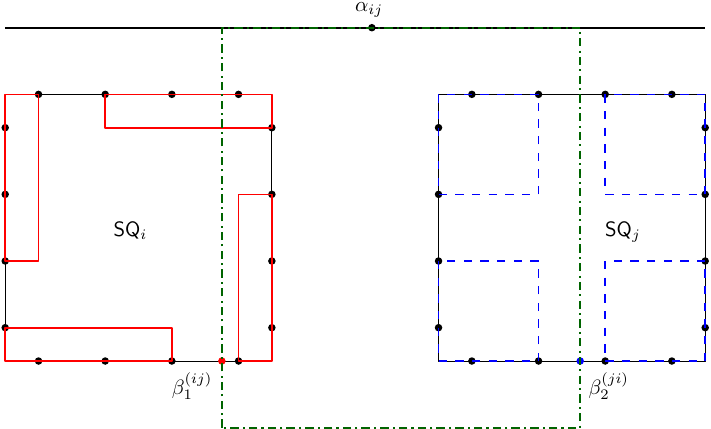}
    \caption{Constraint rectangle covering constraint points and constraint guard point}
    \label{fig:final_rec}
\end{figure}

\end{itemize}

\paragraph*{\underline{\large Parameter:}}

The total number of rectangles allowed in the solution is
\begin{align*}
    k &= 4 \cdot (\text{number of variables}) + (\text{number of constraints})\\
    &= 4 k' + m'\\
    &= 4 k' + \tfrac{3k'}{2} = \tfrac{11k'}{2}
\end{align*}

\begin{theorem}
\bcdaprfull is $\mathrm{W}[1]$-hard parameterized by the number of rectangles.
\end{theorem}

\subsection{Correctness of the Reduction}

\begin{lemma}
The \rcsp instance $I=(Z,D,\mathcal{C})$ is a \yes instance if and only
if \bcdaprfull(\bcdaprshort) is a \yes instance.
\end{lemma}

\begin{proof}
\textbf{\underline{\large Completeness:}}
Let a \rcsp admits a satisfying assignment
$\sigma = (\sigma_1, \dots, \sigma_{k'})$, where $\sigma_i \in D$. For each $i\in [k']$, let $\sigma_i = a_{p_i}$.
We select the four variable rectangles
$BL^{(i)}_{p_i}$, $LT^{(i)}_{p_i}$, $TR^{(i)}_{p_i}$, and $RB^{(i)}_{p_i}$
corresponding to the value $a_{p_i}\in D$ in the variable gadget ${\sf SQ}_i$.

By the placement of the variable points on ${\sf SQ}_i$, these four rectangles
together cover all variable points in ${\sf SQ}_i$. Moreover, by the placement
of the constraint points, there are exactly $3$ constraint points lying between
$a^{(iB)}_{p_i}$ and $a^{(iB)}_{p_i+1}$ on the bottom side of ${\sf SQ}_i$, that are not covered by the selected variable
rectangles.

Now consider any constraint $C_{ij} \in \mathcal{C}$ involving variables
$z_i$ and $z_j$. Since $\sigma$ is a satisfying assignment, the pair
$(\sigma_i, \sigma_j) = (a_{p_i}, a_{p_j}) \in D \times D$ satisfies
$C_{ij}$. By the construction of the constraint gadget, there exists a
constraint rectangle that covers the constraint guard point $\alpha_{ij}$ horizontally
and simultaneously covers the two constraint points
$\beta^{(ij)}_{p_i}$ and $\beta^{(ji)}_{p_j}$ vertically. We select this rectangle in the
solution.

Thus, for each variable $z_i$ we select exactly $4$ variable rectangles, and
for each constraint we select exactly $1$ constraint rectangle. In total, we
select
\[
4k' + m' \;=\; \frac{11k'}{2} \;=\; k
\]
rectangles. Clearly, these rectangles cover all variable points and constraint points on every variable gadget
${\sf SQ}_i$ as well as all constraint guard points, and hence cover all points
in the constructed instance.

\noindent
\textbf{\underline{\large Soundness:}} 
Let $\widehat{R}_1, \dots, \widehat{R}_{k}$ be a solution to \bcdaprshort instance.
We first show that any solution of size $k$ must satisfy certain
properties.

\begin{itemize}

    \item Consider any variable gadget ${\sf SQ}_i$. We need to select at least $4$ variable rectangles in ${\sf SQ}_i$ since those rectangles are the only ones that cover the variable points. Moreover, to cover all the variable points in ${\sf SQ}_i$ using $4$ variable rectangles one has to choose one rectangle of each type, namely $\{BL^{(i)}_{*}, LT^{(i)}_{*}, TR^{(i)}_{*}, RB^{(i)}_{*}\}$, where $*\in [t]$. Furthermore, the value of $*$ must be same for all $4$ selected variable rectangles, that is $\{BL^{(i)}_{j}, LT^{(i)}_{j}, TR^{(i)}_{j}, RB^{(i)}_{j}\}$, for same $j\in [t]$. As argued above, covering all variable points in a
single gadget requires at least four variable rectangles. Therefore, any
$k$ sized solution must select at least $4k'$ variable rectangles in total. This implies that at  $m' = 3k'/2$ constraint rectangles are in the solution.

 \item Note that every constraint guard point $\alpha_{ij}$ can only be covered by some rectangle whose top horizontal side lies on the unique line $h_{ij}$. Moreover, by construction, since the lines $h_{ij}$ are different for each of the constrain guard points, this implies that, any solution of size $k$ must pick at least $m' = 3k'/2$ constraint rectangles. Thus covering all constraint guard points
requires selecting at least $m'$ constraint rectangles. Hence any feasible
solution must contain at least $4k + m'$ rectangles overall. Since the budget
of the solution is exactly $k = 4k' + m'$, it follows that any $k$ size solution must contain exactly $4k' + m'$ rectangles.
\end{itemize}

Now let $\widehat{R}_1, \dots, \widehat{R}_{4k'}$ denote the
collections of variable rectangles selected in the solution, where we can group them to set of $4$ rectangles, namely $\widetilde{\mathcal{R}}_1, \dots, \widetilde{\mathcal{R}}_{k'}$ where for each
$i \in [k]$, 
\[
\widetilde{\mathcal{R}}_i^
= \left\{BL^{(i)}_{p},\, LT^{(i)}_{p},\, TR^{(i)}_{p},\, RB^{(i)}_{p}\right\}
\]
for some $p \in [t]$. In other words, for every variable gadget ${\sf SQ}_i$,
the solution selects exactly four variable rectangles corresponding to a single
domain value $a_p \in D$. Now we construct an assignment for {\sc $3$-regular $2$-CSP} as follows: for every variable $z_i$, let $\pi(i) \in [t]$ be the index of the variable rectangles that the solution picked to cover the variable points in ${\sf SQ}_i$. We assign $z_i = a_{\pi(i)}$. Now we show that $(a_{\pi(1)}, \dots, a_{\pi(k')})$ is a valid assignment, that is, every constraint is satisfied. Consider any constraint $C_{ij}$ involving the two variables $z_i$ and $z_j$. By the way of assignment, we know that the variable rectangles $\left\{BL^{(i)}_{\pi(i)},\, LT^{(i)}_{\pi(i)},\, TR^{(i)}_{\pi(i)},\, RB^{(i)}_{\pi(i)}\right\}$ and $\left\{BL^{(j)}_{\pi(j)},\, LT^{(j)}_{\pi(j)},\, TR^{(j)}_{\pi(j)},\, RB^{(j)}_{\pi(j)}\right\}$ are in the solution. So the uncovered constraint points in ${\sf SQ}_i$ lie in between $a_{\pi(i)}^{iB}$ and $a_{\pi(i)+1}^{iB}$ and those in ${\sf SQ}_j$ lie in between $a_{\pi(j)}^{jB}$ and $a_{\pi(j)+1}^{jB}$. Each of these points must be covered vertically by constraint rectangles. By construction, there is exactly one constraint rectangle corresponding to $C_{ij}$ that passes through $\beta^{(ij)}_{\pi_i}$ and $\beta^{(ji)}_{\pi_j}$ vertically. Thus this constraint rectangle must be picked in the solution to cover these two points. 
Moreover, this rectangle also covers the constraint guard point $\alpha_{ij}$ horizontally.  
Thus by construction of the constraint rectangles we get that the assignment $z_i=a_{\pi(i)}$ and $z_j=a_{\pi(j)}$ satisfies the constraint $C_{ij}$.

Since this argument holds for every constraint in $\mathcal{C}$, the
assignment $(a_{\pi(1)}, \dots, a_{\pi(k')})$satisfies all constraints of the \rcsp instance.
Therefore, the original \rcsp instance is satisfiable.
\end{proof}

\section{NP-hardness for {\sc L-Shape Cover} and \prbcshort}
\label{sec:npHardness}
In this section we prove that \plscfull(\plscshort) and \prbcshort\ are NP-complete. Our reductions proceed in two steps. First, we establish NP-hardness of \plscshort\ via a reduction from \pcbvcfull~\cite{kuo2007efficient}. Next, we prove NP-completeness of \prbcshort\ by reducing from a restricted variant of \prbcshort. Recall the definition of axis-aligned L-shape in \Cref{def:lshape}.

\subsection{Reduction to \plscshort}
\label{sec:reduction}
In this subsection we give reduction from \pcbvcshort to \plscshort. 

\paragraph*{\underline{\Large{Construction.}}} Fix an instance $(G=(A, B,E),k_{A},k_{B})$ with
$A=\{a_1,\dots,a_m\}$ and $B=\{b_1,\dots,b_n\}$.
We construct $(P,k)$ where $k \coloneqq k_{A}+k_{B}$ and set $M \coloneqq k+1$. Here $P$ is the set of points. 

\medskip
\noindent\textbf{Edge points (central grid).}
For each edge $a_ib_j\in E$, let the point $p_{ij}\coloneqq(i,j)\in\mathbb{Z}^2$.
Let $P_{\mathrm{edge}}$ denote the set of all such points.

\medskip
\noindent\textbf{Top-left column guards (forcing $k_{B}$ vertical arms).}
Let
$X^{\mathrm{TL}} \coloneqq \{-1,-2,\dots,-k_{B}\}.$
For each $x\in X^{\mathrm{TL}}$, let
\[
Q_x \coloneqq \{(x,n+1),(x,n+2),\dots,(x,n+M)\}.
\]
be a set of $M$ points.
Define $P_{\mathrm{TL}} \coloneqq \bigcup_{x\in X^{\mathrm{TL}}} Q_x$.

\medskip
\noindent\textbf{Bottom-right row guards (forcing $k_{A}$ horizontal arms).}
Let $Y^{\mathrm{BR}} \coloneqq \{-1,-2,\dots,-k_{A}\}.$ 
For each $y\in Y^{\mathrm{BR}}$, let
\[
Q_y \coloneqq \{(m+1,y),(m+2,y),\dots,(m+M,y)\}.
\]
be a set of $M$ points.
Define $P_{\mathrm{BR}} \coloneqq \bigcup_{y\in Y^{\mathrm{BR}}} Q_y$.


\medskip
\noindent\textbf{Final point set.}
Define $P \coloneqq P_{\mathrm{edge}}\cup P_{\mathrm{TL}}\cup P_{\mathrm{BR}}$.

\begin{remark}[guards only on designated lines]
We place guard points only on the $k_{B}$ designated columns in $X^{\mathrm{TL}}$ and the $k_{A}$ designated rows in $Y^{\mathrm{BR}}$.
\end{remark}

A schematic geometric overview of our reduction can be found in Figure~\ref{fig:lsc-overview}. 

\begin{figure}[t]
\centering
\begin{tikzpicture}[scale=0.75, every node/.style={font=\small}]
  \def\m{6}
  \def\n{4}
  \def\kA{2}
  \def\kB{2}

  \pgfmathtruncatemacro{\mone}{\m-1}
  \pgfmathtruncatemacro{\none}{\n-1}

  \draw[thick] (0,0) rectangle (\m,\n);
  \node[above right] at (0,0) {$[m]\times[n]$};

  \foreach \x in {1,...,\mone} { \draw[gray!35] (\x,0) -- (\x,\n); }
  \foreach \y in {1,...,\none} { \draw[gray!35] (0,\y) -- (\m,\y); }

  \fill (1,1) circle (1.5pt);
  \fill (4,3) circle (1.5pt);
  \fill (5,2) circle (1.5pt);
  \node[right] at (6,2) {$P_{\mathrm{edge}}$};

  \draw[thick, dashed] (-\kB-1,\n+1) rectangle (-0.4,\n+3.2);
  \node[above] at (-1.2,\n+3.2) {$P_{\mathrm{TL}}$};
  \foreach \x in {-1,-2} {
    \draw[gray!60] (\x,\n+1.05) -- (\x,\n+3.15);
    \foreach \t in {0,0.4,0.8,1.2} { \fill (\x,\n+1.15+\t) circle (1.3pt); }
  }

  \draw[thick, dashed] (\m+1,-\kA-1) rectangle (\m+3.2,-0.4);
  \node[below] at (\m+2.0,-\kA-1) {$P_{\mathrm{BR}}$};
  \foreach \y in {-1,-2} {
    \draw[gray!60] (\m+1.05,\y) -- (\m+3.15,\y);
    \foreach \t in {0,0.4,0.8,1.2} { \fill (\m+1.15+\t,\y) circle (1.3pt); }
  }

  \draw[very thick] (-1,\n+1.05) -- (-1,\n+3.10);
  \draw[very thick] (-1,2) -- (\m,2);
  \fill (-1,2) circle (2pt);

  \draw[very thick] (3,-1) -- (\m+3.10,-1);
  \draw[very thick] (3,-1) -- (3,\n);
  \fill (3,-1) circle (2pt);

\end{tikzpicture}
\caption{Schematic view of the construction: $P_{\mathrm{edge}}$ encodes edges; $P_{\mathrm{TL}}$ forces $k_{B}$ vertical arms on the designated top-left columns; $P_{\mathrm{BR}}$ forces $k_{A}$ horizontal arms on the designated bottom-right rows.}
\label{fig:lsc-overview}
\end{figure}
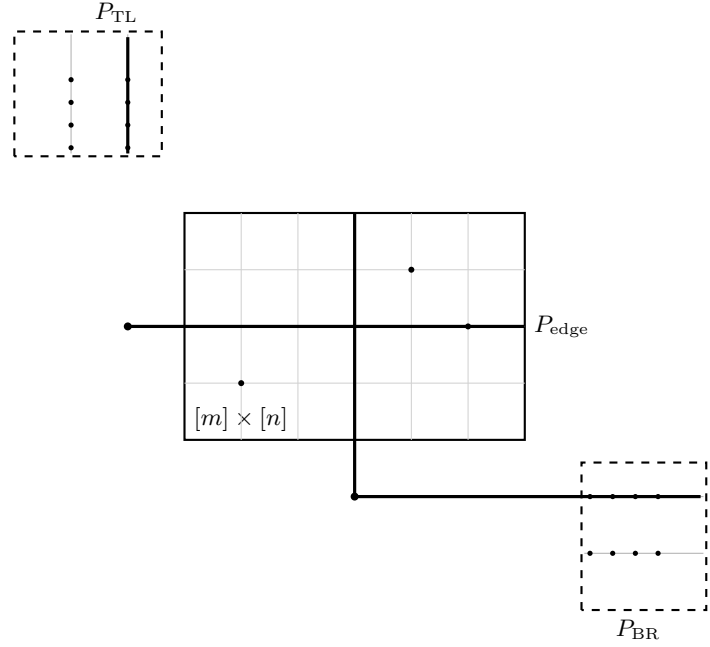
\subsection{Forcing lemmas}

\begin{lemma}[Top-left forces vertical arms on $X^{\mathrm{TL}}$]
\label{lem:TL-forces-vertical}
Let $\mathcal{L}$ be a family of at most $k$ L-shapes covering $P$.
Then for each $x\in X^{\mathrm{TL}}$, there exists an L-shape in $\mathcal{L}$ whose vertical arm lies on the line $\{(x,y):y\in\mathbb{R}\}$.
Moreover, the $k_{B}$ columns in $X^{\mathrm{TL}}$ require $k_{B}$ \emph{distinct} L-shapes.
\end{lemma}
\begin{proof}
Fix $x\in X^{\mathrm{TL}}$. The $M$ guard points on that column have distinct $y$-coordinates.
A horizontal segment lies on a single $y$ and hence can cover at most one of them.
If no L-shape uses a vertical arm on $x$, then these $M=k+1$ points require more than $k$ distinct L-shapes, contradicting $|\mathcal{L}|\le k$.
Thus some L-shape must place its vertical arm on $x$.

Finally, one L-shape has its vertical arm on a single $x$-coordinate, so it cannot cover all the guard points on two different columns.
Hence the $k_{B}$ columns require $k_{B}$ distinct L-shapes.

\end{proof}

\begin{lemma}[Bottom-right forces horizontal arms on $Y^{\mathrm{BR}}$]
\label{lem:BR-forces-horizontal}
Let $\mathcal{L}$ be a family of at most $k$ L-shapes covering $P$.
Then for each $y\in Y^{\mathrm{BR}}$, there exists an L-shape in $\mathcal{L}$ whose horizontal arm lies on the line $\{(x,y):x\in\mathbb{R}\}$.
Moreover, the $k_{A}$ rows in $Y^{\mathrm{BR}}$ require $k_{A}$ \emph{distinct} L-shapes.
\end{lemma}
\begin{proof}
Fix $y\in Y^{\mathrm{BR}}$. The $M$ guard points on that row have distinct $x$-coordinates.
A vertical segment lies on a single $x$ and hence can cover at most one of them.
If no L-shape uses a horizontal arm on $y$, then these $M=k+1$ points require more than $k$ distinct L-shapes, contradicting $|\mathcal{L}|\le k$.
Thus some L-shape must place its horizontal arm on $y$.

As one L-shape has its horizontal arm on a single $y$-coordinate, it cannot cover all the guard points on two different rows.
Hence the $k_{A}$ rows require $k_{A}$ distinct L-shapes.
\end{proof}

\begin{corollary}[Counting available central-grid arms]
\label{cor:free-arms}
Let $\mathcal{L}$ be a cover of $P$ by at most $k$ L-shapes.
Then:
\begin{itemize}
\item at most $k-k_{B} = k_{A}$ shapes in $\mathcal{L}$ can have their vertical arm on a line $x=i$ with $i\in[m]$, and
\item at most $k-k_{A} = k_{B}$ shapes in $\mathcal{L}$ can have their horizontal arm on a line $y=j$ with $j\in[n]$.
\end{itemize}
\end{corollary}
\begin{proof}
By Lemma~\ref{lem:TL-forces-vertical}, there are $k_{B}$ distinct shapes whose vertical arms are committed to the guard columns $x\in X^{\mathrm{TL}}\subseteq \mathbb{Z}_{<0}$.
Thus among at most $k$ shapes, at most $k-k_{B}=k_{A}$ shapes can have vertical arms on any other $x$-coordinate, in particular on $x=i$ for $i\in[m]$.
The horizontal statement follows symmetrically from Lemma~\ref{lem:BR-forces-horizontal}.
\end{proof}





\begin{remark}[overlaps are allowed]
An L-shape may simultaneously hit one top-left guard column and one bottom-right guard row; we do not forbid this.
\end{remark}
\subsection{Correctness}

\begin{lemma}[Completeness]
\label{lem:completeness}
If the \pcbvcshort\ instance is a \yes-instance, then the constructed \plscshort\ instance $(P,k)$ is a \yes-instance.
\end{lemma}
\begin{proof}
Let $A'\subseteq A$, $B'\subseteq B$ with $|A'|\le k_{A}$, $|B'|\le k_{B}$ be a constrained vertex cover.
We build a family of $k=k_{A}+k_{B}$ L-shapes covering $P$.

\smallskip
\noindent\emph{Bottom-right shapes (encode choices in $A$).}
For each $y\in Y^{\mathrm{BR}}$ (there are $k_{A}$ such rows), choose any $a_i\in A'$ (with repetition allowed if $|A'|<k_{A}$),
and place an L-shape with corner at $(i,y)$ whose horizontal arm extends right to $x=m+M$ (covering $(m+1,y),\dots,(m+M,y)$)
and whose vertical arm extends up through the central grid (lying on $x=i$).

\smallskip
\noindent\emph{Top-left shapes (encode choices in $B$).}
For each $x\in X^{\mathrm{TL}}$ (there are $k_{B}$ such columns), choose any $b_j\in B'$ (with repetition allowed if $|B'|<k_{B}$),
and place an L-shape with corner at $(x,j)$ whose vertical arm extends up to $y=n+M$ (covering $(x,n+1),\dots,(x,n+M)$)
and whose horizontal arm extends right through the central grid (lying on $y=j$).

All guard points are covered by construction. Consider any edge-point $(i,j)\in P_{\mathrm{edge}}$ (i.e., $a_ib_j\in E$).
Since $A'\cup B'$ is a vertex cover, either $a_i\in A'$ or $b_j\in B'$.
In the first case, some bottom-right L-shape has vertical arm on $x=i$ and covers $(i,j)$; in the second case, some top-left L-shape has horizontal arm on $y=j$ and covers $(i,j)$.
Hence all points in $P$ are covered using $k$ L-shapes.
\end{proof}

\begin{lemma}[Soundness]
\label{lem:soundness}
If the constructed \plscshort\ instance $(P,k)$ is a \yes-instance, then the original \pcbvcshort\ instance is a \yes-instance.
\end{lemma}
\begin{proof}
Let $\mathcal{L}$ be a family of at most $k$ L-shapes covering $P$.

Define
\[
\begin{aligned}
A' &\coloneqq \Bigl\{a_i \in A \;:\; 
   \exists\, L \in \mathcal{L} \text{ whose vertical arm lies on } x=i 
   \text{ and covers some point of } P_{\mathrm{edge}} \Bigr\},\\[4pt]
B' &\coloneqq \Bigl\{b_j \in B \;:\; 
   \exists\, L \in \mathcal{L} \text{ whose horizontal arm lies on } y=j 
   \text{ and covers some point of } P_{\mathrm{edge}} \Bigr\}.
\end{aligned}
\]

\smallskip
\noindent\emph{Budgets.}
By Corollary~\ref{cor:free-arms}, at most $k_{A}$ shapes can have vertical arms on lines $x=i$ with $i\in[m]$,
hence $|A'|\le k_{A}$. Similarly, $|B'|\le k_{B}$.

\smallskip
\noindent\emph{Edge coverage.}
Let $a_ib_j\in E$ and consider its point $(i,j)\in P_{\mathrm{edge}}$.
This point is covered by some $L\in\mathcal{L}$, and hence lies on the vertical arm (implying $a_i\in A'$) or on the horizontal arm (implying $b_j\in B'$). 
Thus $(A',B')$ is a valid constrained vertex cover.
\end{proof}

\begin{theorem}
\plscshort\ is NP-complete.
\end{theorem}
\begin{proof}
Membership in NP was observed earlier. NP-hardness follows from the polynomial-time reduction in Section~\ref{sec:reduction} together with Lemmas~\ref{lem:completeness} and~\ref{lem:soundness}.
\end{proof}

\subsection{Orientations and restriction to $\NE$-oriented L-shapes}
\label{sec:ne-restriction}
For our final reduction, we will use the fact that our construction so far employs only a single orientation of the L-shape. We therefore begin by defining all possible orientations of an axis-aligned L-shape.

\begin{definition}[Oriented L-shapes]
Let $c=(x_c,y_c)$ be the corner and let $L = V\cup H$ be an L-shape with corner $c$.
We say $L$ is
\begin{itemize}
\item \textbf{$\NE$-oriented} if the vertical arm has $c$ as its \emph{lower} endpoint and the horizontal arm has $c$ as its \emph{left} endpoint;
\item \textbf{$\NW$-oriented} if the vertical arm has $c$ as its \emph{lower} endpoint and the horizontal arm has $c$ as its \emph{right} endpoint;
\item \textbf{$\SE$-oriented} if the vertical arm has $c$ as its \emph{upper} endpoint and the horizontal arm has $c$ as its \emph{left} endpoint;
\item \textbf{$\SW$-oriented} if the vertical arm has $c$ as its \emph{upper} endpoint and the horizontal arm has $c$ as its \emph{right} endpoint.
\end{itemize}
See Figure~\ref{fig:orientations}. 
\end{definition}

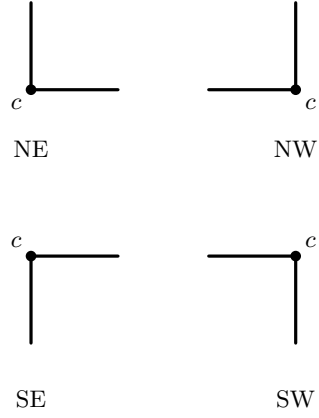
\begin{figure}[t]
\centering
\begin{tikzpicture}[scale=1.0, line cap=round, line join=round, every node/.style={font=\small}]
  \def\arm{1.15}
  \def\sep{3.5}

  \begin{scope}[shift={(0,0)}]
    \draw[very thick] (0,0) -- (0,\arm);
    \draw[very thick] (0,0) -- (\arm,0);
    \fill (0,0) circle (2pt);
    \node[below left] at (0,0) {$c$};
    \node[below] at (0,-0.55) {$\NE$};
  \end{scope}

  \begin{scope}[shift={(\sep,0)}]
    \draw[very thick] (0,0) -- (0,\arm);
    \draw[very thick] (0,0) -- (-\arm,0);
    \fill (0,0) circle (2pt);
    \node[below right] at (0,0) {$c$};
    \node[below] at (0,-0.55) {$\NW$};
  \end{scope}

  \begin{scope}[shift={(0,-2.2)}]
    \draw[very thick] (0,0) -- (0,-\arm);
    \draw[very thick] (0,0) -- (\arm,0);
    \fill (0,0) circle (2pt);
    \node[above left] at (0,0) {$c$};
    \node[below] at (0,-1.65) {$\SE$};
  \end{scope}

  \begin{scope}[shift={(\sep,-2.2)}]
    \draw[very thick] (0,0) -- (0,-\arm);
    \draw[very thick] (0,0) -- (-\arm,0);
    \fill (0,0) circle (2pt);
    \node[above right] at (0,0) {$c$};
    \node[below] at (0,-1.65) {$\SW$};
  \end{scope}
\end{tikzpicture}
\caption{The four orientations of axis-aligned L-shapes.}
\label{fig:orientations}
\end{figure}

\paragraph*{Restriction to $\NE$-oriented L-shapes}
We define \pnelscfull\ (\pnelscshort) as a specialization of \plscshort\ in which we are allowed to use only \emph{\(\NE\)-oriented} axis-aligned L-shapes.  Formally, given a point set \(P\) and an integer \(k\), the question is whether there exists \(k' \le k\) such that \(P\) can be covered by \(k'\) \(\NE\)-oriented axis-aligned L-shapes.

Our reduction already produces \(\NE\)-solutions. In the completeness construction (Lemma~\ref{lem:completeness}), every L-shape we place is \(\NE\)-oriented: each top-left shape has a vertical arm extending upward to reach \(y>n\) and a horizontal arm extending rightward into the grid, while each bottom-right shape has a horizontal arm extending rightward to reach \(x>m\) and a vertical arm extending upward into the grid. Therefore, we obtain the following.

\begin{theorem}
\pnelscshort\ is NP-complete.
\end{theorem}

\subsection{Reduction from \pnelscshort}
\label{subsec:brc-reduction}

In this section we prove that \prbcshort\ is NP-complete. We give a reduction from \pnelscshort, which is NP-complete
(Section~\ref{sec:ne-restriction}). Let $(P_0,k)$ be an instance of \pnelscshort. By applying an affine scaling and translation,
we may assume that $P_0 \subseteq [m]\times[n]$ for some integers $m,n\ge 1$.

Set $M\coloneqq2k+1$. We construct a point set
\[
P \;\coloneqq\; P_0 \ \cup\ P^{\mathrm{V}}_{\mathrm{guard}} \ \cup\ P^{\mathrm{H}}_{\mathrm{guard}}
\]
and ask whether $P$ can be boundary-covered by at most $k$ axis-parallel rectangles.

\paragraph*{Vertical guard columns (to the right, but below all horizontal guards).}
For each $t\in[k]$, define
\[
X_t \coloneqq m+M+t
\qquad\text{and}\qquad
V_t \coloneqq \{(X_t,n+1),(X_t,n+2),\dots,(X_t,n+M)\}.
\]
Put $P_{\mathrm{guard}}^{V} \coloneqq \bigcup_{t=1}^{k} V_t$. 

\paragraph*{Horizontal guard rows (strictly above the vertical guards, and strictly left of them).}
For each $s\in[k]$, define
\[
Y_s \coloneqq n+M+s
\qquad\text{and}\qquad
H_s \coloneqq \{(m+1,Y_s),(m+2,Y_s),\dots,(m+M,Y_s)\}.
\]
Put $P_{\mathrm{guard}}^{H} \coloneqq \bigcup_{s=1}^{k} H_s$. 

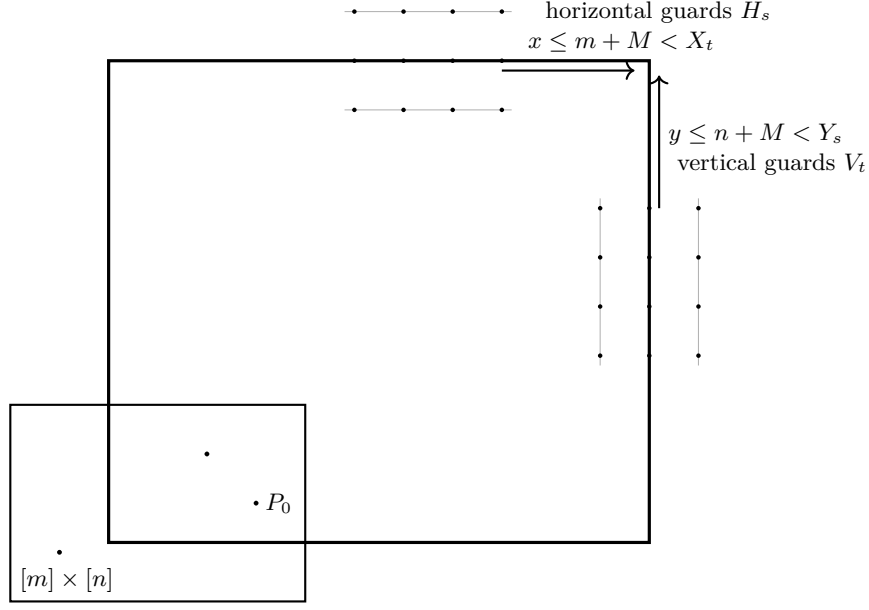
\begin{figure}[t]
\centering
\begin{tikzpicture}[scale=0.65, every node/.style={font=\small}]
  \def\m{6}
  \def\n{4}
  \def\k{3}
  \def\M{4} 

  \pgfmathtruncatemacro{\xL}{\m+1}
  \pgfmathtruncatemacro{\xR}{\m+\M}
  \pgfmathtruncatemacro{\xVG}{\m+\M+2}
  \pgfmathtruncatemacro{\yB}{\n+1}
  \pgfmathtruncatemacro{\yT}{\n+\M}
  \pgfmathtruncatemacro{\yHG}{\n+\M+2}

  \draw[thick] (0,0) rectangle (\m,\n);
  \node[above right] at (0,0) {$[m]\times[n]$};

  \fill (1,1) circle (1.4pt);
  \fill (4,3) circle (1.4pt);
  \fill (5,2) circle (1.4pt);
  \node[right] at (5,2) {$P_0$};

  \foreach \t in {0,1,2}{
    \pgfmathtruncatemacro{\xx}{\xVG+\t}
    \draw[gray!60] (\xx,\yB-0.2) -- (\xx,\yT+0.2);
    \foreach \yy in {1,...,\M}{ \fill (\xx,\n+\yy) circle (1.3pt); }
  }
  \node[above] at (\xVG+3.5,\yT+0.5) {vertical guards $V_t$};

  \foreach \s in {0,1,2}{
    \pgfmathtruncatemacro{\yy}{\yHG+\s}
    \draw[gray!60] (\xL-0.2,\yy) -- (\xR+0.2,\yy);
    \foreach \xx in {1,...,\M}{ \fill (\m+\xx,\yy) circle (1.3pt); }
  }
  \node[right] at (\xR+0.7,\yHG+2) {horizontal guards $H_s$};

  \pgfmathtruncatemacro{\xRect}{\xVG+1}
  \pgfmathtruncatemacro{\yRect}{\yHG+1}
  \draw[very thick] (2,1.2) rectangle (\xRect,\yRect);

  \draw[->, thick] (\xR,\yRect-0.2) -- (\xRect-0.3,\yRect-0.2);
  \node[above] at ({(\xR+\xRect)/1.85},\yRect) {$x\le m+M < X_t$};

  \draw[->, thick] (\xRect+0.2,\yT) -- (\xRect+0.2,\yRect-0.3);
  \node[right] at (\xRect+0.2,{(\yT+\yRect)/2}) {$y\le n+M < Y_s$};
\end{tikzpicture}

\caption{Schematic placement of guards: vertical guard columns lie strictly to the right of $x=m+M$, while horizontal guard rows lie strictly above $y=n+M$ and entirely within the strip $m+1\le x\le m+M$. Thus the two guard regions are disjoint.}
\label{fig:brc-guards}
\end{figure}
\paragraph*{Placement invariant.}
Every point of each $H_s$ satisfies $x\le m+M$, while each vertical guard column is at $x=X_t\ge m+M+1$; hence all horizontal guards lie strictly to the left of all vertical guards.
Similarly, every point of each $V_t$ satisfies $y\le n+M$, while each horizontal guard row is at $y=Y_s\ge n+M+1$; hence all vertical guards lie strictly below all horizontal guards
(see Figure~\ref{fig:brc-guards}).

\medskip
We now prove the key structural claim used in the soundness direction: in any solution of size at most $k$, every rectangle must ``pay for'' exactly two guard sets (and hence, in particular, it cannot be a rectangle that only serves $P_0$).

\begin{lemma}[Guard-side forcing and counting]
\label{lem:brc-structure}
Let $M\coloneqq2k+1$, and let $\mathcal{R}$ be a family of at most $k$ axis-parallel rectangles such that
\(
P \subseteq \bigcup_{R\in\mathcal{R}} \bd(R)
\),
where $P=P_0\cup\bigcup_{t\in[k]}V_t\cup\bigcup_{s\in[k]}H_s$ in the instance constructed above.

For a rectangle $R$ and $t\in[k]$, we say that $R$ \emph{serves} $V_t$ if
 $x=X_t$ is a vertical side of $R$ and 
$\bd(R)$ contains at least three point of $V_t$. Similarly, $R$ \emph{serves} $H_s$ if  $y=Y_s$ is a horizontal side of $R$ and 
$\bd(R)$ contains at least three point of $H_s$.

Then the following statements hold.
\begin{enumerate}
\item Every vertical guard set $V_t$ is served by some rectangle in $\mathcal{R}$, and every horizontal guard set $H_s$
is served by some rectangle in $\mathcal{R}$.
\item Each rectangle in $\mathcal{R}$ serves points from \emph{at most two} guard sets among
$\{V_1,\dots,V_k, H_1,\dots,$ $H_k\}$.
Moreover, if a rectangle serves two vertical guard sets, then it serves no horizontal guard set; and symmetrically,
if it serves two horizontal guard sets, then it serves no vertical guard set.
\item Consequently, $\lvert\mathcal{R}\rvert=k$, and every rectangle in $\mathcal{R}$ serves \emph{exactly two} guard sets.
In particular, no rectangle in $\mathcal{R}$ serves fewer than two guard sets.
\end{enumerate}
\end{lemma}

\begin{proof}
We repeatedly use the following fact.

\smallskip
\noindent\textbf{Fact.}
Fix a vertical line $\ell:x=c$. If a rectangle $R$ does \emph{not} have a vertical side on $\ell$, then
$\bd(R)\cap \ell$ consists of at most two points (where $\ell$ meets the top and bottom sides of $R$).
An analogous statement holds for a horizontal line.

\smallskip
\noindent\emph{Proof of (1).}
Fix $t\in[k]$. Suppose no rectangle in $\mathcal{R}$ has a vertical side on $x=X_t$.
By the Fact, each rectangle then covers at most $2$ points of $V_t$, so all rectangles together cover at most
$2|\mathcal{R}|\le 2k$ points of $V_t$. But $|V_t|=M=2k+1$, a contradiction.
Hence some rectangle, say $R\in\mathcal{R}$, has a vertical side on $x=X_t$ and $\bd(R)$ contains at least three point of $V_t$, that is, serves $V_t$.
The argument for each $H_s$ is symmetric.

\smallskip
\noindent\emph{Proof of (2).}
A rectangle has exactly two vertical sides and two horizontal sides; hence it can serve at most two distinct sets
among $\{V_1,\dots,V_k\}$ and at most two distinct sets among $\{H_1,\dots,H_k\}$.

Now assume $R$ serves two vertical guard sets, say $V_{t_1}$ and $V_{t_2}$ with $t_1\neq t_2$.
Then the two vertical sides of $R$ lie on $x=X_{t_1}$ and $x=X_{t_2}$, and therefore every boundary point of $R$
has
\[
x \;\ge\; \min\{X_{t_1},X_{t_2}\}\;\ge\; m+M+1.
\]
However, every point in any horizontal guard set $H_s$ has $x\in\{m+1,\dots,m+M\}$, so $\bd(R)$ cannot contain
any point of any $H_s$. Thus $R$ serves no horizontal guard set. The horizontal--vertical symmetric statement is analogous.
Therefore each rectangle serves points from at most two guard sets in total.

\smallskip
\noindent\emph{Proof of (3).}
By (1), all $2k$ guard sets must be served by rectangles in $\mathcal{R}$.
By (2), each rectangle serves at most two guard sets.
Hence at least $k$ rectangles are necessary to serve all $2k$ guard sets.
Since $|\mathcal{R}|\le k$, we conclude $|\mathcal{R}|=k$ and that every rectangle serves exactly two guard sets.
\end{proof}

\begin{lemma}[Completeness]
\label{lem:brc-complete}
If $(P_0,k)$ is a \yes-instance of \pnelscshort, then the constructed instance $(P,k)$ is a \yes-instance of \prbcshort.
\end{lemma}
\begin{proof}
Let $L_1,\dots,L_{k'}$ be $\NE$-oriented L-shapes covering $P_0$, where $k'\le k$.
Write the corner of $L_i$ as $c_i=(\alpha_i,\beta_i)$. Since $P_0\subseteq [m]\times[n]$ and each $L_i$ is $\NE$-oriented,
we may assume (by extending arms if needed) that $L_i$ covers points of $P_0$ only along the ray on the line $x=\alpha_i$ above $y=\beta_i$
and along the ray on the line $y=\beta_i$ to the right of $x=\alpha_i$.

We now build at most $k$ rectangles whose boundaries cover $P$.

\smallskip
\noindent\textbf{Step 1: assign each L-shape to a guard pair.}
Choose any injective map $\pi:[k']\to[k]$ and set $t_i\coloneqq\pi(i)$.
Also choose any injective map $\sigma:[k']\to[k]$ and set $s_i\coloneqq\sigma(i)$.
(If $k'<k$, leave some indices of $[k]$ unused; they will be handled in Step~2.)

\smallskip
\noindent\textbf{Step 2: create rectangles.}
For each $i\in[k']$, define a rectangle $R_i$ by:
\[
\text{left side } x=\alpha_i,\qquad \text{bottom side } y=\beta_i,\qquad \text{right side } x=X_{t_i},\qquad \text{top side } y=Y_{s_i}.
\]
By construction, $\bd(R_i)$ contains the left and bottom sides of $R_i$, hence contains $L_i$ (as a subset),
and also contains the full guard lines $x=X_{t_i}$ and $y=Y_{s_i}$.

If $k'<k$, for each remaining index $t\in[k]\setminus \{t_i:i\in[k']\}$ pick an unused $s\in[k]\setminus \{s_i:i\in[k']\}$
(there are equally many of each), and create an auxiliary rectangle with, say, left side $x=0$, bottom side $y=0$, right side $x=X_t$,
and top side $y=Y_s$. These auxiliary rectangles only help cover the remaining guard sets and are irrelevant to covering $P_0$.

\smallskip
\noindent\textbf{Step 3: coverage.}
Each rectangle $R_i$ covers all points of $V_{t_i}$ on its right side $x=X_{t_i}$, and all points of $H_{s_i}$ on its top side $y=Y_{s_i}$.
Thus all guard points are covered.

Now consider any point $p\in P_0$. Since $L_1,\dots,L_{k'}$ cover $P_0$, the point $p$ lies on $L_i$ for some $i$.
As $L_i\subseteq \bd(R_i)$, we have $p\in \bd(R_i)$. Hence $p$ is boundary-covered by our rectangles.

Therefore $(P,k)$ is a \yes-instance of \prbcshort.
\end{proof}


\begin{lemma}[Soundness: extracting $\NE$-L-shapes]
\label{lem:brc-sound-shape}
Let $\mathcal{R}$ be a family of at most $k$ rectangles whose boundaries cover $P$.
Then $P_0$ can be covered by at most $k$ $\NE$-oriented L-shapes.
\end{lemma}

\begin{proof}
Let
\[
\mathcal{R}_0 \;\coloneqq\; \{\,R\in \mathcal{R} : \bd(R) \cap P_0 \neq \emptyset\,\}.
\]
Clearly $|\mathcal{R}_0|\le |\mathcal{R}|\le k$.

Fix any rectangle $R\in\mathcal{R}_0$. By Lemma~\ref{lem:brc-structure}(3), every rectangle in $\mathcal{R}$
serves exactly two guard sets. We claim that the two guard sets served by $R$ must be \emph{one} vertical guard set
and \emph{one} horizontal guard set.

Indeed, if $R$ served two vertical guard sets, then by Lemma~\ref{lem:brc-structure}(2) every boundary point of $R$
would have $x\ge m+M+1$, whereas every point of $P_0\subseteq[m]\times[n]$ has $x\le m$, contradicting
$\bd(R)\cap P_0\neq\emptyset$. Symmetrically, $R$ cannot serve two horizontal guard sets. Hence there exist
indices $t(R)\in[k]$ and $s(R)\in[k]$ such that $R$ serves $V_{t(R)}$ and serves $H_{s(R)}$.

We next claim that for such an $R$, the side on $x=X_{t(R)}$ is the \emph{right} side of $R$ and the side on
$y=Y_{s(R)}$ is the \emph{top} side of $R$. Since $R\in\mathcal{R}_0$, there exists a point of $P_0$ on $\bd(R)$,
hence $\bd(R)$ contains a point with $x\le m$ and a point with $y\le n$. Thus the left side of $R$ has
$x\le m$ and the bottom side of $R$ has $y\le n$. But $X_{t(R)}\ge m+M+1>m$ and $Y_{s(R)}\ge n+M+1>n$,
so $x=X_{t(R)}$ cannot be the left side and $y=Y_{s(R)}$ cannot be the bottom side; therefore they are the right and top sides,
respectively.

For each rectangle $R\in\mathcal{R}_0$, define $L(R)$ to be the union of the \emph{left} and \emph{bottom} sides of $R$.
Then $L(R)$ is an $\NE$-oriented L-shape.

Now fix any point $p\in P_0$. Since $\mathcal{R}$ boundary-covers $P$, the point $p$ lies on $\bd(R)$ for some
$R\in\mathcal{R}$, and hence for some $R\in\mathcal{R}_0$ by definition of $\mathcal{R}_0$.
For this rectangle $R$, we have shown that its right side is at $x=X_{t(R)}>m$ and its top side is at $y=Y_{s(R)}>n$.
Because $p$ satisfies $x(p)\le m$ and $y(p)\le n$, it cannot lie on the right or top side of $R$.
Therefore $p$ lies on the left side or the bottom side of $R$, i.e., $p\in L(R)$.

Hence the family $\{\,L(R): R\in\mathcal{R}_0\,\}$ covers $P_0$ and has size at most $|\mathcal{R}_0|\le k$.
\end{proof}

\begin{theorem}
\label{thm:brc-npc}
\prbcshort\ is NP-complete.
\end{theorem}
\begin{proof}
Membership in NP is immediate: a certificate consists of the corner coordinates of at most $k$ rectangles, and one can verify in polynomial time
that every input point lies on the boundary of at least one rectangle.

For NP-hardness, we use the reduction from \pnelscshort\ described in Section~\ref{subsec:brc-reduction}.
Lemma~\ref{lem:brc-complete} proves completeness of the reduction.
For soundness, Lemma~\ref{lem:brc-sound-shape} shows how to extract from any boundary-cover of $P$ by at most $k$ rectangles
a cover of $P_0$ by at most $k$ $\NE$-L-shapes.
Thus $(P_0,k)$ is a \yes-instance of \pnelscshort\ if and only if $(P,k)$ is a \yes-instance of \prbcshort,
and the reduction is polynomial-time.
\end{proof}



\section{An FPT algorithm for \prbcshort}
\label{sec:fpt-algorithm}
In this section we present an FPT algorithm for \prbcshort\ parameterized by $k$.
The algorithm follows the outline from Section~\ref{sec:overview}; each step is implemented and analyzed in its own subsection below.
As a result, \prbcshort\ is fixed-parameter tractable with respect to $k$.

\subsection{Discretization of coordinates}
We begin by \emph{discretizing the plane} using the coordinate values induced by the input point set
$P=\{p_1,\dots,p_n\}\subseteq \mathbb{R}^2$, where $p_i=(x_i,y_i)$ for each $i\in[n]$.
Let $X$ and $Y$ denote the sets of distinct $x$- and $y$-coordinates appearing in $P$, respectively:
\[
X := \{\,x_i \mid (x_i,y_i)\in P\,\}
\qquad\text{and}\qquad
Y := \{\,y_i \mid (x_i,y_i)\in P\,\}.
\]
Let $\gridpts_X=\langle \alpha_1,\alpha_2,\dots,\alpha_{N_x}\rangle$ be the increasing ordering of $X$,
and let $\gridpts_Y=\langle \beta_1,\beta_2,\dots,\beta_{N_y}\rangle$ be the increasing ordering of $Y$,
where $N_x=|X|\le n$ and $N_y=|Y|\le n$.
Thus $\alpha_1<\alpha_2<\cdots<\alpha_{N_x}$ and $\beta_1<\beta_2<\cdots<\beta_{N_y}$.

\noindent

\paragraph*{Discretization.}\label{def:discretization}
Given a point set $P\subseteq \mathbb{R}^2$, let
$\gridpts_X=\langle \alpha_1,\dots,\alpha_{N_x}\rangle$ and
$\gridpts_Y=\langle \beta_1,\dots,\beta_{N_y}\rangle$
be the increasing orderings of the distinct $x$- and $y$-coordinates appearing in $P$.
We draw the vertical lines $x=\alpha_j$ for all $j\in[N_x]$ and the horizontal lines $y=\beta_\ell$ for all $\ell\in[N_y]$.
These lines form an axis-parallel grid whose \emph{grid points} are the intersection points
\[
(\alpha_j,\beta_\ell)\qquad\text{for all } j\in[N_x]\text{ and }\ell\in[N_y].
\]
By construction, every point of $P$ lies on a grid point of this grid (indeed, on an intersection of one drawn vertical and one drawn horizontal line). We now prove the main lemma of this section, showing that we may assume, without loss of generality, that every rectangle in an optimal solution has all four sides aligned with grid lines. This yields the following lemma.

\begin{lemma}[Grid-aligned solution for rectangles]
\label{lem:rect-discretization}
Let $(P,k)$ be an instance of \prbcshort and let
$X=\{x(p):p\in P\}$ and $Y=\{y(p):p\in P\}$.
Let $\gridpts_X=\langle \alpha_1<\alpha_2<\cdots<\alpha_{N_x}\rangle$ and
$\gridpts_Y=\langle \beta_1<\beta_2<\cdots<\beta_{N_y}\rangle$ be the sorted lists of distinct coordinates in $X$ and $Y$.
If $(P,k)$ is a \yes-instance, then there exists a solution $\mathcal{R}$ of size at most $k$ such that for every rectangle
$R\in\mathcal{R}$, all four sides of $R$ lie on grid lines of the form $x=\alpha_j$ and $y=\beta_\ell$ (equivalently, all four
corners of $R$ are grid points $(\alpha_j,\beta_\ell)$).
\end{lemma}

\begin{proof}
Let $\mathcal{F}$ be the set of all feasible solutions of size at most $k$, i.e., families $\cR$ of at most $k$
rectangles such that $P\subseteq \bigcup_{R\in\cR}\partial R$.
Fix a solution $\cR\in\mathcal{F}$ that lexicographically minimizes the pair
\[
\Phi(\cR)\ :=\ \Bigl(\#\mathrm{NG}(\cR),\ \sum_{R\in\cR}\bigl(|a_R|+|b_R|+|c_R|+|d_R|\bigr)\Bigr),
\]
where $\#\mathrm{NG}(\cR)$ denotes the total number of sides of rectangles in $\cR$ that are \emph{not} aligned with
any grid line, and where each rectangle $R$ is written as
$R=\{(x,y)\in\mathbb{R}^2: a_R\le x\le b_R,\ c_R\le y\le d_R\}$.

We claim that $\#\mathrm{NG}(\cR)=0$, which implies the desired grid-alignment.
Assume for contradiction that $\#\mathrm{NG}(\cR)>0$, and choose a rectangle
$R=\{(x,y): a\le x\le b,\ c\le y\le d\}\in\cR$ having a non-grid side.
By symmetry, assume that the left side $x=a$ is not a grid line; in particular, $a\notin X$.
Let $S:=\{y : (a,y)\in P\}$ be the set of $y$-coordinates of points of $P$ on this side.
Since $a\notin X$, we have $S=\emptyset$: no point of $P$ lies on the line $x=a$.

Let $\alpha_s:=\max\{\alpha\in X : \alpha<a\}$ and $\alpha_{s+1}:=\min\{\alpha\in X : \alpha>a\}$.
(If one of these does not exist, interpret it as $-\infty$ or $+\infty$; the argument below uses only the fact that
there is no $x$-coordinate of a point of $P$ in the open interval between $a$ and the chosen grid line.)
Since $a\notin X$, at least one of $\alpha_s$ or $\alpha_{s+1}$ is strictly closer to $a$ than any other grid line.
Pick $\alpha^\star\in\{\alpha_s,\alpha_{s+1}\}$ such that $|a-\alpha^\star|$ is minimized (break ties arbitrarily),
and define the modified rectangle
\[
R^\star\ :=\ \{(x,y)\in\mathbb{R}^2 : \alpha^\star \le x \le b,\ c\le y\le d\}.
\]
Let $\cR^\star := (\cR\setminus\{R\})\cup\{R^\star\}$.

\smallskip
\noindent\emph{Feasibility is preserved.}
We show that $\cR^\star$ still boundary-covers $P$.
Let $p\in P$ be any point that is covered by $\partial R$ and consider how it can lie on $\partial R$.
Since $a\notin X$, no point of $P$ lies on the side $x=a$ (by $S=\emptyset$), so $p$ must lie on one of the other
three sides: $x=b$, $y=c$, or $y=d$.

If $p$ lies on $x=b$, then it is still on $x=b$, hence still on $\partial R^\star$.
If $p$ lies on $y=c$ (respectively $y=d$), then $p=(x,c)$ (resp.\ $(x,d)$) with $a\le x\le b$.
Because there is no $x$-coordinate of any point of $P$ strictly between $a$ and $\alpha^\star$, this $x$ also satisfies
$\alpha^\star \le x \le b$, and hence $p$ still lies on the bottom (resp.\ top) side of $R^\star$.
Therefore every point of $P$ that was covered \emph{only} by $\partial R$ remains covered by $\partial R^\star$, and all
other points are unaffected; thus $\cR^\star\in\mathcal{F}$.

\smallskip
\noindent\emph{The potential strictly decreases.}
The rectangle $R^\star$ has the same three sides $x=b$, $y=c$, $y=d$ as $R$, but its left side is now $x=\alpha^\star$,
which \emph{is} a grid line. Hence $\#\mathrm{NG}(\cR^\star) < \#\mathrm{NG}(\cR)$, so
$\Phi(\cR^\star) < \Phi(\cR)$ lexicographically, contradicting the choice of $\cR$.

This contradiction shows that $\#\mathrm{NG}(\cR)=0$, i.e., every side of every rectangle in $\cR$ lies on a grid line.
Since grid lines intersect only at grid points, all corners are grid points as claimed.
\end{proof}

In what follows we work entirely with such grid-aligned rectangles and squares.

\subsection{Covering by few lines via Vertex Cover}
We now prove that if there is a solution with at most $k$ rectangles, then the points in $P$ can be covered by at most $4k$ grid lines (vertical or horizontal). 


Consider the grid induced by the coordinate sets $X$ and $Y$ from the discretization step (see Section~\ref{def:discretization}). Construct a bipartite graph $G = (V_{\sf vert} \uplus V_{\sf hor}, E)$ where
\begin{itemize}
  \item $V_{\sf vert}$ has one vertex $v_j$ for each distinct vertical line $x = \lambda_j$, i.e., $V_{\sf vert} = \{v_j \mid \lambda_j \in X\}$.
  \item $V_{\sf hor}$ has one vertex $h_\ell$ for each distinct horizontal line $y = \gamma_\ell$, i.e., $V_{\sf hor} = \{h_\ell \mid \gamma_\ell \in Y\}$.
  \item The edge set $E(G)$ consists of edges $e_p = (v_j, h_\ell)$ for each point $p = (\lambda_j, \gamma_\ell) \in P$.
\end{itemize}

A subset $C \subseteq V_{\sf vert} \uplus V_{\sf hor}$ is a \emph{vertex cover} of $G$ if every edge $e_p \in E(G)$ has at least one endpoint in $C$. In the next lemma, we relate the solution to \prbcshort to a vertex cover of $G$.

\begin{lemma}[Few-line cover of points]
\label{lem:line-cover}
Let $(P,k)$ be an instance of \prbcshort. If $(P,k)$ has a solution of size at most $k$ , then the bipartite graph $G$ has a vertex cover $C$ of size at most $4k$. Moreover, such a minimum vertex cover can be found in polynomial time.

\end{lemma}

\begin{proof}
Let $\mathcal{R}= \{R_1,\dots,R_t\}$, $t \le k$ be a solution to \prbcshort. For each rectangle in $\cR$, it has has two vertical and two horizontal boundary lines. For $G$, we construct the set $C\subseteq V(G)$ as follows: for each rectangle $R_i, i \in [t]$, insert the vertices corresponding to its two vertical and two horizontal boundary lines into $C$. Clearly $|C| \le 4t \le 4k$. Now we show that $C$ is a vertex cover for $G$.

Now fix any point $p = (\lambda_j, \gamma_\ell) \in P$. Since $\cR$ is a solution to \prbcshort, the point $p$ lies on the boundary of some rectangle $R_i \in \cR$. If $p$ happens to be a corner of $R_i$, we break ties by treating it as lying on the horizontal side of $R_i$.
Now either
\begin{itemize}
  \item $p$ lies on a vertical side (left or right boundary) of $R_i$, in which case $v_j \in C$, or
  \item $p$ lies on a horizontal side (top or bottom boundary) of $R_i$, in which case $h_\ell \in C$.
\end{itemize}
Further by the construction, $G$ only has edges corresponding to the points in $P$. Thus we have that for each edge $e_p=(v_j,h_\ell) \in E(G)$, at least one endpoint is in $C$~\cite{korte2008combinatorial}. Thus $C$ is a vertex cover.

We know that, for bipartite graphs, a minimum vertex cover can be computed in polynomial time using K\H{o}nig's theorem together with a maximum matching algorithm. Thus, we can compute a minimum vertex cover $C^\star$ and check whether $|C^\star| \le 4k$.

If $|C^\star| > 4k$, then by the pigeonhole principle, at least one of $|C^\star \cap V_{\sf vert}|$ or $|C^\star \cap V_{\sf hor}|$ must exceed $2k$. Since these vertices correspond to vertical and horizontal lines in $\mathbb{R}^2$, covering all points in $P$ would then require at least $2k{+}1$ horizontal or vertical sides. However, any solution using at most $k$ rectangles has at most $2k$ horizontal and at most $2k$ vertical sides. Therefore, no solution of size at most $k$ can exist in this case.
\end{proof}

\noindent
\textbf{Solution Supporting Important Lines:} Let $C$ be a vertex cover of $G$ obtained via \Cref{lem:line-cover}, with $|C| \le 4k$.  
By the construction of $G$, the set $C$ corresponds to at most $4k$ vertical or horizontal lines that together cover all points in $P$.  
We denote this collection of lines by $\impline$. Clearly $|\impline| \leq 4k$.  
Since the vertex cover $C$ can be computed algorithmically, the set $\impline$ can be computed as well~\cite{korte2008combinatorial}.

We remark that this step is not new: it is exactly the classical polynomial-time algorithm for computing a minimum-size set of axis-parallel lines that covers a given point set in the plane (see, e.g.,~\cite{gaur2007covering,DBLP:journals/dam/HassinM91}).


\subsection{Exceptional points: definition and basic bounds}


Let $\impline_{\hor}$ and $\impline_{\ver}$ denote the sets of horizontal and vertical lines in $\impline$, respectively.
Since $\impline=\impline_{\hor}\uplus \impline_{\ver}$ and $|\impline|\le 4k$, we have
$|\impline_{\hor}|\le 4k$ and $|\impline_{\ver}|\le 4k$.

We now state a few basic properties of the lines in $\impline$. We present them for horizontal lines; the vertical case is
completely analogous. Fix a line $L\in \impline_{\hor}$, and let $P_L := P\cap L = \{p_1,\ldots,p_t\}$,
where the points are ordered by increasing $x$-coordinate. Writing $p_j=(x_j,y_L)$, we have $
x_1 < x_2 < \cdots < x_t$.  A rectangle $R$ can interact with $L$ in only two relevant ways:
either one of its horizontal sides lies on $L$, or $L$ cuts through the interior of $R$ (in which case $L$ meets $\bd(R)$
in exactly two points).

\begin{lemma}[Intersection of a horizontal line with a rectangle]
\label{lem:line-rect-intersection}
Let $L$ be a horizontal line $y=y_L$ and let $R=[a,b]\times[c,d]$ be an axis-parallel rectangle with $a<b$ and $c<d$.
Then exactly one of the following holds.
\begin{enumerate}
\item $\bd(R)\cap L=\emptyset$ (this happens when $y_L\notin[c,d]$);
\item $L$ coincides with a horizontal side of $R$ (i.e., $y_L\in\{c,d\}$), and then $\bd(R)\cap L = [a,b]\times\{y_L\}$;
\item $c<y_L<d$, and then $\bd(R)\cap L=\{(a,y_L),(b,y_L)\}$.
\end{enumerate}
In particular, if no horizontal side of $R$ lies on $L$, then $|\bd(R)\cap L|\le 2$.
\end{lemma}

\begin{proof}
If $y_L\notin[c,d]$, then $L$ lies strictly above or below $R$ and hence does not meet $\bd(R)$.
If $y_L=c$ or $y_L=d$, then $L$ coincides with the bottom or top side of $R$, and the intersection is the entire segment
$[a,b]\times\{y_L\}$.
Finally, if $c<y_L<d$, then $L$ crosses the interior of $R$ and meets the boundary exactly at the two vertical sides,
namely at $(a,y_L)$ and $(b,y_L)$.
\end{proof}

\medskip
This yields the promised ``two modes'' for covering points on $L$ (aligned vs. crossing); see Figure~\ref{fig:aligned-crossing}.

\begin{enumerate}
\item \textbf{Aligned coverage:} 
A point $p \in P_L$ is said to be \emph{aligned covered} by a rectangle $R$ if the line $L$
coincides with either the top or the bottom side of $R$, and the point $p$ lies on that
side of $R$.

\item \textbf{Crossing coverage:} 
A point $p \in P_L$ is said to be \emph{crossing covered} by a rectangle $R$ if $p$ lies on
either the left or the right side of $R$ and $R$ intersects the line $L$ at
exactly two points where $p$ is one of them. 

\end{enumerate}

\begin{definition}[Exceptional points on a line]
\label{def:exceptional-point}
Let $\mathcal{R}$ be a family of rectangles. A point $p\in P_L$ is \emph{exceptional (w.r.t.\ $\mathcal{R}$)} if it is covered by
some rectangle $R\in\mathcal{R}$ that has no horizontal side on $L$.
\end{definition}

\begin{lemma}[Few exceptional points per line]
\label{lem:few-exceptional}
Let $\mathcal{R}$ be a family of at most $k$ rectangles. For any horizontal line $L$,
the number of exceptional points of $P_L$ (w.r.t.\ $\mathcal{R}$) is at most $2k$.
\end{lemma}

\begin{proof}
By Lemma~\ref{lem:line-rect-intersection}, any rectangle $R$ with no horizontal side on $L$ satisfies
$|\bd(R)\cap L|\le 2$. Hence such a rectangle can cover at most two points of $P_L$ in the crossing mode.
Summing over at most $k$ rectangles gives at most $2k$ exceptional points on $L$.
\end{proof}

Every other point on $L$, that is, every point of $P_L$ that is \emph{not} exceptional, must be covered in the aligned mode:
namely, it must lie on the top or bottom side of some rectangle whose horizontal side is contained in $L$.
In particular, non-exceptional points may appear in arbitrarily large numbers along $L$, but they can only be covered by rectangles
that have a horizontal side on $L$. All definitions and lemmas stated above for horizontal lines extend verbatim to vertical lines in $\impline_{\ver}$
(by swapping the roles of $x$ and $y$, and ``horizontal'' with ``vertical'').

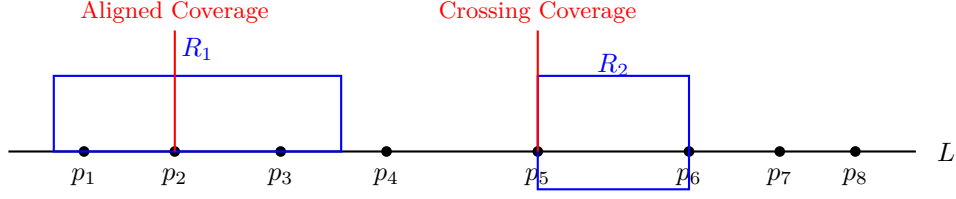
\begin{figure}[t]
  \centering
  \begin{tikzpicture}[scale=1.0]

    \draw[thick] (0,0) -- (12,0);
    \node at (12.4,0) {$L$};

    \foreach \x/\lab in {
        1/1, 2.2/2, 3.6/3, 5.0/4,
        7.0/5, 9.0/6,
        10.2/7, 11.2/8
    } {
      \fill (\x,0) circle (2pt);
      \node[below] at (\x,-0.1) {$p_{\lab}$};
    }

    \draw[blue,thick]
      (0.6,0) -- (4.4,0) -- (4.4,1) -- (0.6,1) -- cycle;
    \node[blue] at (2.5,1.35) {$R_1$};

    \draw[blue,thick]
      (7,-0.5) -- (9,-0.5) -- (9,1.0) -- (7,1.0) -- cycle;
    \node[blue] at (8,1.15) {$R_2$};

    \draw[red,thick] (7,0) -- (7,1.6);
    \node[red] at (7,1.85) {\small Crossing Coverage};

    \draw[red,thick] (2.2,0) -- (2.2,1.6);

    \node[red] at (2.2,1.85) {\small Aligned Coverage};

  \end{tikzpicture}
  \caption{
    $\{p_1, \dots, p_8\}$ are the points along a horizontal line $L$. Coverage of $p_1, p_2, p_3$ by $R_1$ is example of aligned coverage where as coverage of $p_5, p_6$ by $R_2$ is crossing coverage.
  }
  \label{fig:aligned-crossing}
\end{figure}

\subsection{Skeletons}
\label{subsec:skeletons}

Recall that by \Cref{lem:line-cover} we can compute, in polynomial  time, a set of at most $4k$ grid-aligned lines $\impline \;=\; \impline_{\ver}\ \uplus\ \impline_{\hor}$  
such that every point of $P$ lies on at least one line of $\impline$.

Intuitively, a solution by at most $k$ rectangles must ``use'' many of these lines as rectangle sides; otherwise too many points
on a line would have to be covered by rectangles that merely \emph{cross} the line, which is impossible with only $k$ rectangles.
We capture this alignment information by the notion of a \emph{skeleton}.


Let $\cR=\{R_1,\dots,R_k\}$ be a family of axis-parallel rectangles, where $R_i \;=\; [\ell_i,r_i]\times[b_i,t_i]$,
and denote its four sides by $R_i(L)$, $R_i(R)$, $R_i(B)$, and $R_i(T)$ (left, right, bottom, top).

\paragraph*{Supporting line of a side.}
Since each side $S$ of an axis-parallel rectangle is a (closed) axis-parallel line segment, it lies on a unique
axis-parallel line. We call this line the \emph{supporting line} of $S$ and denote it by $\suppline(S)$.
Formally,
\[
\suppline\bigl(R_i(L)\bigr)=\{(x,y)\in\mathbb{R}^2 \mid x=\ell_i\},\qquad
\suppline\bigl(R_i(R)\bigr)=\{(x,y)\in\mathbb{R}^2 \mid x=r_i\},
\]
\[
\suppline\bigl(R_i(B)\bigr)=\{(x,y)\in\mathbb{R}^2 \mid y=b_i\},\qquad
\suppline\bigl(R_i(T)\bigr)=\{(x,y)\in\mathbb{R}^2 \mid y=t_i\}.
\]
In particular, $\suppline(S)$ is vertical for $S\in\{R_i(L),R_i(R)\}$ and horizontal for
$S\in\{R_i(B),R_i(T)\}$.


\begin{definition}[Skeleton]
\label{def:skeleton}
A \emph{skeleton} is a function
\[
\sigma:\ \{(i,\side)\mid i\in[k],\ \side\in\{L,R,B,T\}\}\ \to\ \impline \cup \{\bot\}
\]
such that:
\begin{enumerate}
\item if $\sigma(i,L) \neq \bot $ then $\sigma(i,L) \in \impline_{\ver}$ (and similarly  if $\sigma(i,R) \neq \bot $ then $\sigma(i,R) \in \impline_{\ver}$);

\item if $\sigma(i,B) \neq \bot $ then $\sigma(i,B) \in \impline_{\hor}$ (and similarly  if $\sigma(i,T) \neq \bot $ then $\sigma(i,T) \in \impline_{\hor}$);

\item (No side-identification) $\sigma(i,L)\neq \sigma(i,R)$ whenever both are lines in $\impline_{\ver}$, and
      $\sigma(i,B)\neq \sigma(i,T)$ whenever both are lines in $\impline_{\hor}$.
\item If $\sigma(i,\side)=\bot$, it means that the corresponding side of $R_i$ is \emph{not} aligned with any line of $\impline$.
\end{enumerate}

\end{definition}

\begin{definition}[Compatibility]
\label{def:compatible-skeleton}
Let $\cR=\{R_1,\dots,R_k\}$ be a family of rectangles. We say that $\cR$ is \emph{compatible} with a skeleton $\sigma$
if for every $i\in[k]$ and every $\side\in\{L,R,B,T\}$:
\begin{itemize}
\item if $\sigma(i,\side)\in \impline$, then $\suppline(R_i(\side))=\sigma(i,\side)$;
\item if $\sigma(i,\side)=\bot$, then $\suppline(R_i(\side))\notin \impline$.
\end{itemize}
\end{definition}


We next bound the number of skeletons that the algorithm needs to enumerate.

\begin{lemma}[Number of skeletons]
\label{lem:number-of-skeletons}
Let $|\impline|\le 4k$. The number of skeletons (as in \Cref{def:skeleton}) is at most $(4k+1)^{4k}=k^{\Oh(k)}$.
\end{lemma}

\begin{proof}
There are $4k$ sides in total (four sides for each of $k$ rectangles).
For each side, the skeleton chooses either a specific line from $\impline$ or the special symbol $\bot$.
Thus there are at most $(|\impline|+1)^{4k}\le (4k+1)^{4k}$ assignments.

The additional consistency conditions in \Cref{def:skeleton} only \emph{restrict} assignments, so this is an upper bound
on the number of valid skeletons.
Finally, $(4k+1)^{4k}=k^{\Oh(k)}$.
\end{proof}

We now justify that skeleton guessing is without loss of generality: every solution induces a skeleton.

\begin{lemma}[Every solution admits a compatible skeleton]
\label{lem:solution-induces-skeleton}
Let $\cR=\{R_1,\dots,R_k\}$ be any family of  axis parallel rectangles.
Then there exists a skeleton $\sigma$ such that $\cR$ is compatible with $\sigma$.
\end{lemma}

\begin{proof}
Define $\sigma$ side-by-side.
Fix $i\in[k]$ and $\side\in\{L,R,B,T\}$, and let $\ell=\suppline(R_i(\side))$ be the supporting line of that side.

If $\ell\in\impline$, set $\sigma(i,\side):=\ell$; otherwise set $\sigma(i,\side):=\bot$.
This definition ensures the compatibility conditions of \Cref{def:compatible-skeleton} by construction.

Moreover, $\ell$ is vertical for $\side\in\{L,R\}$ and horizontal for $\side\in\{B,T\}$, so the type-consistency
requirements in \Cref{def:skeleton} are satisfied.
Finally, a rectangle cannot have its left and right sides on the same vertical line, and cannot have its bottom and top sides
on the same horizontal line; hence the non-identification condition in \Cref{def:skeleton}(3) holds whenever both sides
are mapped to lines of $\impline$.
\end{proof}

\begin{lemma}[Enumerating all valid skeletons]
\label{lem:enumerate-skeletons}
There exists an algorithm that, given $(P,k)$ and the line family $\impline=\impline_{\ver}\uplus\impline_{\hor}$ with
$|\impline|\le 4k$, enumerates the set of all valid skeletons (as in \Cref{def:skeleton}) in time
\[
k^{\Oh(k)}\cdot |P|^{\Oh(1)}.
\]
Moreover, the number of enumerated skeletons is at most $(4k+1)^{4k}=k^{\Oh(k)}$, and by \Cref{lem:solution-induces-skeleton},
if $(P,k)$ admits a solution by at most $k$ rectangles, then at least one enumerated skeleton is compatible with some solution.
\end{lemma}

\begin{proof}
We enumerate all assignments
\[
\sigma:\ \{(i,\side)\mid i\in[k],\ \side\in\{L,R,B,T\}\}\ \to\ \impline \cup \{\bot\}
\]
and keep only those satisfying the validity conditions in \Cref{def:skeleton}.
There are $4k$ arguments and each has at most $|\impline|+1\le 4k+1$ choices, hence the total number of assignments is at most
$(4k+1)^{4k}=k^{\Oh(k)}$ (cf.\ \Cref{lem:number-of-skeletons}).
Checking whether a candidate $\sigma$ satisfies the conditions of \Cref{def:skeleton} is polynomial in $k$ (and thus dominated by
the enumeration).
Since $\impline$ itself is computed in polynomial time in $|P|$ (by \Cref{lem:line-cover}), the total running time is
$k^{\Oh(k)}\cdot |P|^{\Oh(1)}$.

Finally, \Cref{lem:solution-induces-skeleton} implies that every solution $\cR$ admits a compatible skeleton, and therefore if
$(P,k)$ is a \yes-instance, then at least one enumerated skeleton is compatible with some solution.
\end{proof}

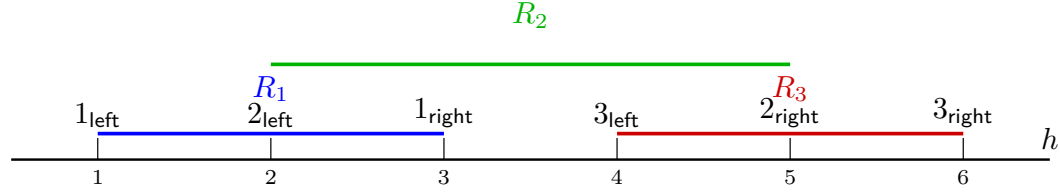
\begin{figure}[t]
  \centering
  \resizebox{\linewidth}{!}{%
  \begin{tikzpicture}[scale=1]

    \draw[thick] (0,0) -- (12,0);
    \node[above] at (12,0) {$h$};

    \coordinate (L1) at (1.0,0);
    \coordinate (L2) at (3.0,0);
    \coordinate (R1) at (5.0,0);
    \coordinate (L3) at (7.0,0);
    \coordinate (R2) at (9.0,0);
    \coordinate (R3) at (11.0,0);

    \foreach \P/\name in {L1/1_{\sf left},L2/2_{\sf left},R1/1_{\sf right},L3/3_{\sf left},R2/2_{\sf right},R3/3_{\sf right}} {
      \draw (\P) -- ++(0,0.25);
      \node[above] at ($( \P ) + (0,0.25)$) {$\name$};
    }

    \draw[blue,very thick] (L1) ++(0,0.3) -- ($(R1)+(0,0.3)$);
    \node[blue,above] at ($(L1)!0.5!(R1)+(0,0.55)$) {$R_1$};

    \draw[green!70!black,very thick] (L2) ++(0,1.1) -- ($(R2)+(0,1.1)$);
    \node[green!70!black,above] at ($(L2)!0.5!(R2)+(0,1.4)$) {$R_2$};

    \draw[red!80!black,very thick] (L3) ++(0,0.3) -- ($(R3)+(0,0.3)$);
    \node[red!80!black,above] at ($(L3)!0.5!(R3)+(0,0.55)$) {$R_3$};

    \foreach \P/\idx in {L1/1,L2/2,R1/3,L3/4,R2/5,R3/6} {
      \node[below] at (\P) {\scriptsize $\idx$};
    }


  \end{tikzpicture}%
  }
 \caption{A horizontal line $h\in\impline_{\hor}$ and three rectangles $R_1,R_2,R_3$ whose top or bottom side lies on $h$ under the fixed skeleton $\sigma$. 
  Each such rectangle induces a segment on $h$ with endpoint symbols $i_{\sf left}$ and $i_{\sf right}$, and
  $\pi(\sigma,h)$ records the left-to-right order of all endpoint symbols, subject to $i_{\sf left}$ appearing before
  $i_{\sf right}$ for every $i\in\indx(\sigma,h)$. For example, here $\pi(\sigma,h)(1_{\sf left})=1 < \pi(\sigma,h)(2_{\sf left})=2 < \pi(\sigma,h)(1_{\sf right})=3$ and
      $\pi(\sigma,h)(3_{\sf left})=4 < \pi(\sigma,h)(2_{\sf right})=5 < \pi(\sigma,h)(3_{\sf right})=6$.
  \label{figOrdering}}
\end{figure}

\subsection{Fixing a skeleton and ordering endpoints}
\label{subsec:fix-skeleton-order}
From now on, we fix a valid skeleton $\sigma$ as in \Cref{def:skeleton}.  Please refer to Figure~\ref{figOrdering} for an illustration. We continue to use the notation
$\sigma(i,L)$, $\sigma(i,R)$, $\sigma(i,B)$, and $\sigma(i,T)$ for the values of the skeleton on the four sides of
rectangle $R_i$ (left, right, bottom, top). In particular,
\[
\sigma(i,L),\sigma(i,R)\in \impline_{\ver}\cup\{\bot\}
\qquad\text{and}\qquad
\sigma(i,B),\sigma(i,T)\in \impline_{\hor}\cup\{\bot\}.
\]

\medskip
\noindent
Fix an arbitrary horizontal line $h\in \impline_{\hor}$.
The skeleton tells us exactly which rectangles have their \emph{bottom} or \emph{top} side aligned with $h$.
We collect their indices as
\[
\indx(\sigma,h)\ :=\ \{\, i\in[k]\mid \sigma(i,B)=h \ \text{or}\ \sigma(i,T)=h \,\},
\qquad
m(\sigma,h)\ :=\ |\indx(\sigma,h)|.
\]

For every $i\in\indx(\sigma,h)$, the intersection of $h$ with the corresponding side of $R_i$ is a horizontal segment and
hence has a \emph{left endpoint} and a \emph{right endpoint}. We represent these endpoints symbolically by introducing two
endpoint labels per such $i$:
\begin{align*}
&\ptsend_{\sf left}(\sigma,h)  \;:=\; \{\, i_{\sf left} \mid i\in \indx(\sigma,h)\,\},\\
&\ptsend_{\sf right}(\sigma,h) \;:=\; \{\, i_{\sf right}\mid i\in \indx(\sigma,h)\,\},\\
&\ptsend(\sigma,h)             \;:=\; \ptsend_{\sf left}(\sigma,h)\ \uplus\ \ptsend_{\sf right}(\sigma,h).
\end{align*}
Clearly, $|\ptsend_{\sf left}(\sigma,h)|=|\ptsend_{\sf right}(\sigma,h)|=m(\sigma,h)$ and
$|\ptsend(\sigma,h)|=2m(\sigma,h)$.
(We treat the symbols $i_{\sf left}$ and $i_{\sf right}$ as abstract labels; later we will guess their geometric order along $h$.)

\begin{definition}[Endpoint ordering on a line]
\label{def:endpoint-order}
An \emph{endpoint ordering} on $h$ (with respect to $\sigma$) is a bijection
\[
\pi(\sigma,h):\ \ptsend(\sigma,h)\ \to\ [\,2m(\sigma,h)\,],
\]
which induces the total order
\begin{equation}\label{eq:endpoint-ordering}
\pi(\sigma, h)^{-1}(1)\ \prec\ \pi(\sigma, h)^{-1}(2)\ \prec\ \cdots\ \prec\ \pi(\sigma, h)^{-1}\!\bigl(2m(\sigma,h)\bigr)
\end{equation}
on the symbols in $\ptsend(\sigma,h)$.
\end{definition}

\begin{definition}[Valid endpoint ordering]
\label{def:valid-endpoint-order}
An endpoint ordering $\pi(\sigma,h)$ is \emph{valid} if for every $i\in\indx(\sigma,h)$ we have
\[
\pi(\sigma,h)(i_{\sf left})\ <\ \pi(\sigma,h)(i_{\sf right}),
\]
that is, the left-endpoint symbol of $i$ appears before its right-endpoint symbol in the order~\eqref{eq:endpoint-ordering}.
\end{definition}

\begin{lemma}[Number of valid endpoint orderings]
\label{lem:number-valid-endptorder}
The number of valid endpoint orderings for $h$ is at most
\[
\bigl(2m(\sigma,h)\bigr)!\ \le\ (2k)!.
\]
\end{lemma}

\begin{proof}
There are $\bigl(2m(\sigma,h)\bigr)!$ total bijections from $\ptsend(\sigma,h)$ to $[\,2m(\sigma,h)\,]$,
and validity only restricts this set. Hence the number of valid orderings is at most
$\bigl(2m(\sigma,h)\bigr)! \le (2k)!$, since $m(\sigma,h)\le k$.
\end{proof}

\noindent
The analogous definitions for a vertical line $v\in \impline_{\ver}$ are obtained by replacing $(B,T,\impline_{\hor},h)$
with $(L,R,\impline_{\ver},v)$, and swapping ``left/right endpoint'' with ``bottom/top endpoint''.
When the skeleton $\sigma$ is clear from the context, we will omit it from the notation (e.g., write $\indx(h)$ and $\ptsend(h)$ instead of $\indx(\sigma,h)$ and $\ptsend(\sigma,h)$).

\subsection{Blocks in an endpoint ordering}
\label{subsec:blocks-endpoint-order}
From now on we fix a skeleton $\sigma$, and for every line in $\impline$ we guess an endpoint ordering.
Fix an arbitrary horizontal line $h\in\impline_{\hor}$ and a guess for a \emph{valid} endpoint ordering of
$\ptsend(\sigma,h)$, induced by the bijection $\pi(\sigma,h)$.
Recall that $m(\sigma,h)=|\indx(\sigma,h)|$ and $|\ptsend(\sigma,h)|=2m(\sigma,h)$.
For readability, in this subsection we write $m:=m(\sigma,h)$ and $\pi:=\pi(\sigma,h)$.

\paragraph*{The Dyck word of an ordering (balanced-parentheses condition).}
As in \Cref{subsec:fix-skeleton-order}, the ordering $\pi$ induces a word
$W=W(\sigma,h)=w_1w_2\cdots w_{2m}\in\{\textsf{left},\textsf{right}\}^{2m}$ by setting
$w_j=\textsf{left}$ if $\pi^{-1}(j)\in\ptsend_{\sf left}(\sigma,h)$ and
$w_j=\textsf{right}$ if $\pi^{-1}(j)\in\ptsend_{\sf right}(\sigma,h)$.
Since $\pi$ is a valid endpoint ordering, every prefix of $W$ contains at least as many \textsf{left}'s as \textsf{right}'s.
Moreover, $W$ contains exactly $m$ occurrences of each symbol.
Thus $W$ is a Dyck word (see~\cite{baez2015dyck} for a quick definition) over $\Sigma=\{\textsf{left},\textsf{right}\}$. We {\em note} that the word induced by \(\pi\) need not be injective; that is, two different orderings may lead to the same word.

\begin{definition}[Complete prefixes]
\label{def:complete-prefix}
For every $j\in[2m]$, let
\[
L(j)\ :=\ \#_{\textsf{left}}\bigl(W[1..j]\bigr)
\qquad\text{and}\qquad
R(j)\ :=\ \#_{\textsf{right}}\bigl(W[1..j]\bigr).
\]
We call the prefix $W[1..j]$ \emph{complete} if $L(j)=R(j)$.
\end{definition}

By validity, for all $j\in[2m]$ we have $L(j)\ge R(j)$, and the full word is complete:
$L(2m)=R(2m)=m$.
Equivalently, $W[1..j]$ is complete iff there is no $i\in\indx(\sigma,h)$ with
$\pi(i_{\sf left})\le j < \pi(i_{\sf right})$.

\begin{definition}[Minimal complete blocks]
\label{def:min-complete-block}
Let $0=j_0< j_1<\cdots< j_q=2m$ be the indices of the complete prefixes of $W$,
defined by letting $j_r$ be the smallest index larger than $j_{r-1}$ with $L(j_r)=R(j_r)$.
For each $r\in[q]$, define
\[
\block_r\ :=\ W[j_{r-1}+1 \, ..\, j_r].
\]
We call $\block_1,\ldots,\block_q$ the \emph{minimal complete blocks} (equivalently, \emph{minimal Dyck blocks}) of the ordering.
\end{definition}

\begin{lemma}[Unique decomposition into minimal Dyck blocks]
\label{lem:unique-min-dyck-blocks}
Let $W\in\{\textsf{left},\textsf{right}\}^{*}$ be a Dyck word, and let
$\block_1,\ldots,\block_q$ be the blocks from \Cref{def:min-complete-block}.
Then:
\begin{enumerate}
\item each $\block_r$ is a nonempty Dyck word and no proper prefix of $\block_r$ is a Dyck word (i.e., $\block_r$ is \emph{minimal});
\item $W=\block_1\block_2\cdots \block_q$;
\item this decomposition is unique among all decompositions of $W$ into minimal Dyck words.
\end{enumerate}
\end{lemma}

\begin{proof}
By definition of $j_r$, the prefix $W[1..j_r]$ is a Dyck word and $j_r$ is the \emph{first} complete index after $j_{r-1}$.
Hence $\block_r=W[j_{r-1}+1..j_r]$ is a nonempty Dyck word. Minimality follows because if $\block_r$ had a proper Dyck prefix,
then $W$ would have a complete prefix strictly between $j_{r-1}$ and $j_r$, contradicting the choice of $j_r$.
The concatenation $W=\block_1\cdots\block_q$ is immediate from $j_q=2m$.
For uniqueness, let $W=Y_1\cdots Y_p$ be any decomposition into minimal Dyck words.
Then $|Y_1|$ must be the first complete index $j_1$ (otherwise $Y_1$ is not Dyck or not minimal), so $Y_1=\block_1$.
Stripping this prefix and iterating yields $p=q$ and $Y_r=\block_r$ for all $r$.
\end{proof}

\paragraph*{Blocks and endpoint symbols.}
Each block $\block_r$ corresponds to a contiguous sub-order of endpoint symbols:
it starts at position $j_{r-1}+1$ and ends at position $j_r$ in the total order induced by $\pi$.
Equivalently, the symbols in block $\block_r$ are exactly
\[
\pi^{-1}(j_{r-1}+1)\ \prec\ \pi^{-1}(j_{r-1}+2)\ \prec\ \cdots\ \prec\ \pi^{-1}(j_r).
\]
In particular, $\pi^{-1}(1)$ is the first symbol of $\block_1$, and $\pi^{-1}(2m)$ is the last symbol of $\block_q$.

\subsection{Gaps in an endpoint ordering}
\label{subsec:gaps-endpoint-order}

Recall that we have fixed a skeleton $\sigma$ and, for each line in $\impline$, we guess a valid endpoint ordering.
Fix a horizontal line $h\in\impline_{\hor}$, and abbreviate $m:=m(\sigma,h)$ and $\pi:=\pi(\sigma,h)$.
Let $W=W(\sigma,h)\in\{\textsf{left},\textsf{right}\}^{2m}$ be the Dyck word induced by~$\pi$ as in the previous subsection,
and let $W=X_1X_2\cdots X_q$ be its (unique) decomposition into minimal Dyck blocks.

\paragraph*{Gaps.}
We define the \emph{gaps} of this decomposition to be the $q{+}1$ regions:
one before $X_1$, one between each consecutive pair $X_r,X_{r+1}$, and one after $X_q$.
Accordingly, we set
\[
\gapnum(\pi)\ :=\ q+1.
\]

\paragraph*{Guessing the number of exceptional points per gap.}
Let $K$ be an integer with $0\le K\le 2k$ (intuitively, $K$ is the total number of exceptional points on $h$).
A \emph{gap assignment} is a function
\begin{equation}\label{eq:sumgaps}
\gapfn:\ [\gapnum(\pi)] \to \{0,1,\dots,K\}
\quad\text{such that}\quad
\sum_{j=1}^{\gapnum(\pi)} \gapfn(j)=K.
\end{equation}

\begin{definition}[Gap vector]
\label{def:gap-distribution}
A \emph{gap vector} (for the fixed $(\sigma,h,\pi)$) is a pair $(K,\gapfn)$ where
$K\in\mathbb{Z}_{\ge 0}$ and $\gapfn$ satisfies~\eqref{eq:sumgaps}.
\end{definition}

\begin{lemma}[Counting gap assignments for fixed $K$]
\label{lem:gapvector-count}
Fix $K\ge 0$. The number of functions $\gapfn:[q+1]\to\{0,\dots,K\}$ satisfying
$\sum_{j=1}^{q+1}\gapfn(j)=K$ is exactly $\binom{K+q}{q}$.
\end{lemma}

\begin{proof}
Such a function $\gapfn$ is equivalently a $(q{+}1)$-tuple of nonnegative integers
$(x_1,\dots,x_{q+1})$ with $\sum_{j=1}^{q+1}x_j=K$.
By the standard stars-and-bars argument, the number of such tuples is $\binom{K+q}{q}$.
\end{proof}

\begin{lemma}[Number of gap vectors]
\label{lem:number-gap-vectors}
For a fixed line $h$, the total number of gap vectors $(K,\gapfn)$ over all choices of $K\in\{0,1,\dots,2k\}$
is at most $2^{\Oh(k)}$.
\end{lemma}

\begin{proof}
We have $0\le K\le 2k$. For each fixed $K$, Lemma~\ref{lem:gapvector-count} gives at most $\binom{K+q}{q}$ choices.
Moreover, $q\le m\le k$ because each minimal Dyck block contains at least one \textsf{left} symbol.
Hence, for every $K\le 2k$,
\[
\binom{K+q}{q}\ \le\ \binom{2k+k}{k}\ =\ \binom{3k}{k}\ \leq  6.75^k.
\]
Summing over the at most $2k{+}1$ possible values of $K$ preserves the bound  $2^{\Oh(k)}$.
\end{proof}

\paragraph*{Interpretation.}
Intuitively, $K$ is the total number of exceptional points on $h$, and $\gapfn(j)$ specifies how many of these
exceptional points are placed in the $j$-th gap (i.e., between two consecutive minimal Dyck blocks, or outside all blocks).
In the next subsection, we guess which rectangle sides (left/right sides of the solution rectangles) cover these exceptional
points.

\subsection{Realizable configurations}
\label{subsec:realizable-config}

Fix a skeleton $\sigma$ and a horizontal line $h\in\impline_{\hor}$.
Let $m:=m(\sigma,h)$ and $\pi:=\pi(\sigma,h)$ be a fixed valid endpoint ordering on $h$ as in
\Cref{subsec:fix-skeleton-order}. Let $W=W(\sigma,h)$ be the Dyck word induced by $\pi$, and let
\[
W = X_1X_2\cdots X_q
\]
be its unique decomposition into minimal Dyck blocks (\Cref{lem:unique-min-dyck-blocks}).
We denote these blocks by $\block_1,\ldots,\block_q$.

\paragraph*{Block endpoints.}
For each $r\in[q]$, let $j_r$ be the position of the last symbol of $\block_r$ in $W$ (equivalently, the last position
of $\block_r$ in the order induced by $\pi$). Define the \emph{block-end symbol}
\[
\blkend_r \;:=\; \pi^{-1}(j_r)\ \in\ \ptsend(\sigma,h).
\]
Then $j_1<\cdots<j_q=2m$, and hence
\[
\pi(\blkend_1)\ <\ \pi(\blkend_2)\ <\ \cdots\ <\ \pi(\blkend_q)=2m.
\]
Moreover, the first symbol of $\block_1$ is $\pi^{-1}(1)$ and the last symbol of $\block_q$ is $\pi^{-1}(2m)$.

\begin{definition}[Realizable configuration on a line]
\label{def:full-line-config}
Let $P_h:=\{p\in P \mid p\in h\}$ be the set of input points on $h$.
Let $\gapnum(\pi):=q+1$, and let $\gapfn:[\gapnum(\pi)]\to\mathbb{Z}_{\ge 0}$ be a gap function.

We say that the pair $(\pi,\gapfn)$ is \emph{realizable} on $h$ if there exists a function
\[
\relz:\ \ptsend(\sigma,h)\ \to\ \gridpts_X
\]
such that the following conditions hold.

\begin{enumerate}
\item \textbf{Order preservation.}
For every $j\in[2m-1]$,
\[
\relz\bigl(\pi^{-1}(j)\bigr)\ \le\ \relz\bigl(\pi^{-1}(j+1)\bigr).
\]

\item \textbf{Gap realization.}
The function $\relz$ realizes $\gapfn$ as follows.

\begin{itemize}
\item (\emph{Gap 1: before the first block})
The number of points of $P_h$ that lie strictly to the left of
$\relz\bigl(\pi^{-1}(1)\bigr)$ on $h$ is exactly $\gapfn(1)$.

\item (\emph{Gaps between blocks})
For each $r\in\{2,\ldots,q\}$, let $j_{r-1}:=\pi(\blkend_{r-1})$.
Then the number of points of $P_h$ that lie strictly between the two $x$-coordinates
\[
\relz(\blkend_{r-1})
\qquad\text{and}\qquad
\relz\bigl(\pi^{-1}(j_{r-1}+1)\bigr)
\]
(on the line $h$) is exactly $\gapfn(r)$.

\item (\emph{Gap $q{+}1$: after the last block})
The number of points of $P_h$ that lie strictly to the right of
$\relz(\blkend_q)$ on $h$ is exactly $\gapfn(q+1)$.
\end{itemize}
\end{enumerate}
\end{definition}


\subsection{Exceptional patterns}
\label{subsec:exceptional-patterns}

Fix a skeleton $\sigma$ and a horizontal line $h\in\impline_{\hor}$.
Fix a valid endpoint ordering $\pi=\pi(\sigma,h)$ of $\ptsend(\sigma,h)$ and let
$W=W(\sigma,h)$ be the induced Dyck word, with minimal Dyck-block decomposition
$W=X_1\cdots X_q$ (cf.~\Cref{lem:unique-min-dyck-blocks}).
Let $\gapnum(\pi):=q+1$ be the number of gaps.

\paragraph*{Exceptional points and the gap vector.}
Recall that a point of $P_h:=P\cap h$ is \emph{exceptional} (w.r.t.\ $h$) if it is covered by a rectangle whose boundary
\emph{does not} have a horizontal side on $h$ (equivalently, the covering occurs via a vertical side crossing $h$).
We fix a gap vector $(K,\gapfn)$, where $K$ is our guess for the number of exceptional points on $h$, and
\[
\gapfn:[\gapnum(\pi)]\to\{0,1,\ldots,K\}
\qquad\text{satisfies}\qquad
\sum_{r=1}^{\gapnum(\pi)} \gapfn(r)=K
\]
(cf.~\Cref{def:gap-distribution}).

\paragraph*{Rectangles that can create exceptional points on $h$.}
Let
\[
\indx(\sigma,h)\ :=\ \{\, i\in[k]\mid \sigma(i,B)=h \ \text{or}\ \sigma(i,T)=h \,\}.
\]
Rectangles indexed by $\indx(\sigma,h)$ have a horizontal side on $h$ and thus cover \emph{non-exceptional} points on $h$
horizontally. Only rectangles with no horizontal side on $h$ can create exceptional points. Define
\[
\overline\indx(\sigma,h)\ :=\ [k]\setminus \indx(\sigma,h).
\]
Among these, some rectangles may cross $h$ (and then their vertical sides intersect $h$). We guess the subset
\[
\indx_{\sf exp}(\sigma,h)\ \subseteq\ \overline\indx(\sigma,h)
\]
consisting of those indices $i$ for which $R_i$ intersects $h$ (equivalently, $b_i< h < t_i$ in coordinates). 

\begin{definition}[Exceptional pattern on $h$]
\label{def:exp-pat}
Fix a gap vector $(K,\gapfn)$ and a guess
$\indx_{\sf exp}(\sigma,h)\subseteq [k]\setminus\indx(\sigma,h)$.
An \emph{exceptional pattern} is a function
\[
\exppat(\sigma,h):\ [K]\ \to\ \indx_{\sf exp}(\sigma,h)\times\{L,R\}.
\]
We interpret $\exppat(\sigma,h)(r)=(i,L)$ (respectively, $(i,R)$) as the declaration that the $r$-th exceptional point on
$h$ from left to right is covered by the left (respectively, right) \emph{vertical} side of rectangle $R_i$.
\end{definition}

\paragraph*{Remark.}
The exceptional pattern only records, for each exceptional point on $h$, \emph{which} rectangle side covers it.
Any additional consistency requirements (e.g., that the chosen side actually intersects $h$ and can cover the claimed points)
will be enforced later by the constraints of the \mtcsp\ instance.




\begin{lemma}[Number of exceptional-pattern choices]
\label{lem:number-exppat}
For fixed $\sigma,h$ and $\pi$, the total number of choices of
\[
(K,\ \indx_{\sf exp}(\sigma,h),\ \exppat(\sigma,h))
\]
over all $K\in\{0,1,\ldots,2k\}$ is at most $k^{\Oh(k)}$.
\end{lemma}

\begin{proof}
First choose $K\in\{0,1,\ldots,2k\}$, giving at most $2k+1$ choices.
Then choose $\indx_{\sf exp}(\sigma,h)\subseteq [k]\setminus \indx(\sigma,h)$, giving at most $2^k$ choices.
Finally, for each of the $K$ exceptional points, choose an index in $\indx_{\sf exp}(\sigma,h)$ and a side in $\{L,R\}$,
giving at most $(2|\indx_{\sf exp}(\sigma,h)|)^K\le (2k)^K\le (2k)^{2k}$ choices.
Thus the total number of possibilities is at most $(2k+1)\cdot 2^k \cdot (2k)^{2k}\ =\ k^{\Oh(k)}$. This concludes the proof. 
\end{proof}

\begin{lemma}[Consistency of gaps and exceptional patterns]
\label{lem:gap-pattern-consistency}
Let $(P,k)$ be a \yes-instance of \prbcshort, and let $\cR=\{R_1,\ldots,R_k\}$ be a grid-aligned solution.
Fix $h\in\impline_{\hor}$ and let $\sigma$ be the skeleton induced by $\cR$.
Then there exist:
\begin{itemize}
\item a valid endpoint ordering $\pi=\pi(\sigma,h)$ of $\ptsend(\sigma,h)$,
\item a gap vector $(K,\gapfn)$ with $K\le 2k$,
\item a set $\indx_{\sf exp}(\sigma,h)\subseteq [k]\setminus \indx(\sigma,h)$, and
\item an exceptional pattern $\exppat(\sigma,h)$,
\end{itemize}
such that:
\begin{enumerate}
\item the pair $(\pi,\gapfn)$ is realizable on $h$ in the sense of \Cref{def:full-line-config}, and
\item $\exppat(\sigma,h)$ assigns each of the $K$ exceptional points on $h$ (from left to right) to an index in
$\indx_{\sf exp}(\sigma,h)$ and one of the sides in $\{L,R\}$.
\end{enumerate}
\end{lemma}

\begin{proof}
For each $i\in\indx(\sigma,h)$, the rectangle $R_i$ has either its bottom or its top side on $h$.
This induces a (possibly degenerate) horizontal segment on $h$, whose endpoints have $x$-coordinates in $\gridpts_X$.
We represent these endpoints by the symbols $i_{\sf left},i_{\sf right}\in\ptsend(\sigma,h)$.
Sorting all these endpoint positions from left to right yields a valid endpoint ordering
$\pi$ of $\ptsend(\sigma,h)$.

Let $K$ be the exact number of exceptional points of $P_h$ with respect to $\cR$.
By the exceptional-point bound on a line (cf.~\ref{lem:few-exceptional}), we have $K\le 2k$.
The solution $\cR$ determines a gap function $\gapfn$ by counting, for each of the $q+1$ gaps defined by the minimal Dyck
blocks of $W(\sigma,h)$, how many exceptional points lie in that gap; thus $(K,\gapfn)$ is a gap vector.

Define $\relz:\ptsend(\sigma,h)\to \gridpts_X$ by mapping each endpoint symbol to its true $x$-coordinate in the solution.
By construction, $\relz$ preserves the total order induced by $\pi$, and the counts of points of $P_h$ lying in the
corresponding gaps agree with $\gapfn$. Therefore $(\pi,\gapfn)$ is realizable on $h$ in the sense of
\Cref{def:full-line-config}.

Finally, let $\indx_{\sf exp}(\sigma,h)\subseteq [k]\setminus\indx(\sigma,h)$ be the indices of rectangles that cover at
least one exceptional point on $h$ via a vertical side.
Ordering the exceptional points from left to right and recording, for each such point, a rectangle index $i\in
\indx_{\sf exp}(\sigma,h)$ and whether the point lies on the left or right vertical side of $R_i$ yields an exceptional
pattern $\exppat(\sigma,h)$ as in \Cref{def:exp-pat}.
\end{proof}

\subsection{Reduction to \ddmtcsp}
\label{subsec:reduction-to-mtcsp}

We now give the full parameterized reduction from \prbcshort\ to \ddmtcsp.
Given an instance $(P,k)$ of \prbcshort, we construct (after a bounded amount of guessing that depends only on $k$)
an equivalent \ddmtcsp\ instance.

\paragraph*{Setup.}
We assume that we have already carried out the discretization (\Cref{lem:rect-discretization}) and computed, via
\Cref{lem:line-cover}, a set of grid-aligned lines
\[
\impline \;=\; \impline_{\ver}\ \uplus\ \impline_{\hor},
\qquad |\impline|\le 4k,
\]
that covers all points of $P$.
Moreover, we work only with grid-aligned rectangles, so every rectangle is of the form
$[\ell,r]\times[b,t]$ with $\ell,r\in\gridpts_X$ and $b,t\in\gridpts_Y$.

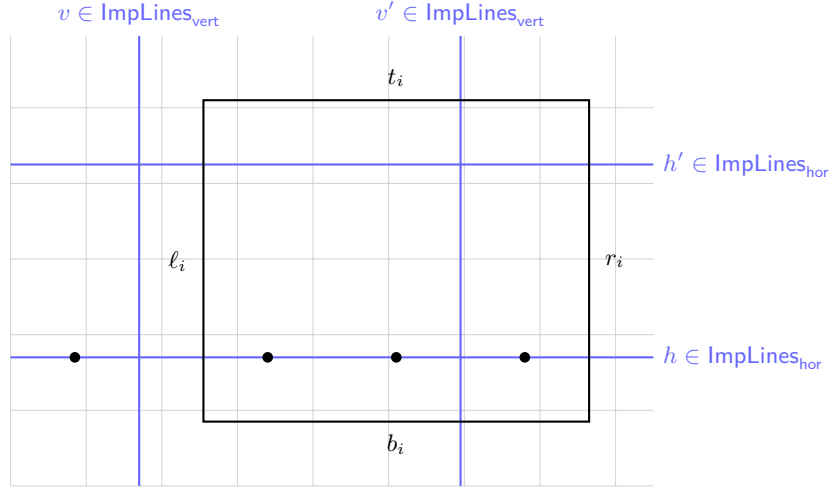
\begin{figure}[t]
\centering
\begin{tikzpicture}[x=0.85cm,y=0.85cm,>=Latex, font=\small]
  \draw[very thin,gray!35] (0,0) grid (10,7);
  \draw[thick,blue!60] (2,0) -- (2,7) node[above] {$v\in\impline_{\ver}$};
  \draw[thick,blue!60] (7,0) -- (7,7) node[above] {$v'\in\impline_{\ver}$};
  \draw[thick,blue!60] (0,2) -- (10,2) node[right] {$h\in\impline_{\hor}$};
  \draw[thick,blue!60] (0,5) -- (10,5) node[right] {$h'\in\impline_{\hor}$};

  \draw[thick] (3,1) rectangle (9,6);
  \node at (6,6.35) {$t_i$};
  \node at (6,0.65) {$b_i$};
  \node at (2.6,3.5) {$\ell_i$};
  \node at (9.4,3.5) {$r_i$};

  \fill (1,2) circle (2pt) node[below] {};
  \fill (4,2) circle (2pt);
  \fill (6,2) circle (2pt);
  \fill (8,2) circle (2pt);

\end{tikzpicture}
\caption{Example picture: grid-aligned lines in $\impline$ (blue),
     one rectangle $R_i=[\ell_i,r_i]\times[b_i,t_i]$, and points on $h$. Global view: after discretization, all points lie on a small set of grid-aligned lines $\impline$, and each rectangle is described by four coordinate variables.}
\label{fig:mtcsp-overview}
\end{figure}

\paragraph*{Canonical ordering of grid coordinates.}
After discretization (\Cref{lem:rect-discretization}), every point of $P$ lies on the grid
$\gridpts_X\times \gridpts_Y$.
We fix once and for all an order-preserving encoding of these grid coordinates into $\mathbb{N}$ as follows.
Let $\gridpts_X=\langle \alpha_1<\cdots<\alpha_{N_x}\rangle$ and $\gridpts_Y=\langle \beta_1<\cdots<\beta_{N_y}\rangle$
be the sorted lists of distinct $x$- and $y$-grid values.
We identify each $\alpha_j$ with the integer $j\in[N_x]$ and each $\beta_\ell$ with the integer $\ell\in[N_y]$.
Equivalently, we map every grid point $(\alpha_j,\beta_\ell)$ to the pair $(j,\ell)\in\mathbb{N}^2$
(lexicographically ordered if needed).

This induces a \emph{canonical} left-to-right order on every horizontal line $h$:
for $p,p'\in P\cap h$ we have $p$ precedes $p'$ iff $x(p)<x(p')$, i.e., iff the integer code of $x(p)$ is smaller.
Similarly, it induces a canonical bottom-to-top order on every vertical line $v$:
for $p,p'\in P\cap v$ we have $p$ precedes $p'$ iff $y(p)<y(p')$.
All successor/predecessor maps and the monotone counting primitives $A_\lambda,B_\lambda$ are henceforth taken with
respect to these fixed global orders (restricted to the relevant finite sets on each line). See Figure~\ref{fig:mtcsp-overview}. 

\subsubsection*{Global guessing pattern}

The reduction proceeds by enumerating the following  (all within $k^{\Oh(k)}$ choices overall):
\begin{description}[leftmargin=2.6em]
\item[(G1)] a skeleton $\sigma$ as in \Cref{def:skeleton};
\item[(G2)] for every line $\lambda\in\impline$, an endpoint ordering $\pi(\sigma,\lambda)$ as in \Cref{def:endpoint-order};
\item[(G3)] for every line $\lambda\in\impline$, a gap vector $(K(\sigma,\lambda),\gapfn(\sigma,\lambda))$ as in \Cref{def:gap-distribution};
\item[(G4)] for every line $\lambda\in\impline$, a set $\indx_{\sf exp}(\sigma,\lambda)\subseteq [k]\setminus \indx(\sigma,\lambda)$ and
an exceptional pattern $\exppat(\sigma,\lambda)$ (without any extra validity condition; the CSP will enforce feasibility).
\end{description}


For each line $\lambda\in\impline$, let $\pi_\lambda:=\pi(\sigma,\lambda)$ be the guessed endpoint ordering, let
$\bigl(K_\lambda,\gapfn_\lambda\bigr):=\bigl(K(\sigma,\lambda),\gapfn(\sigma,\lambda)\bigr)$ be the guessed gap data, and let
\[
\bigl(\indx_{{\sf exp},\lambda},\exppat_\lambda\bigr)
\;:=\;
\bigl(\indx_{\sf exp}(\sigma,\lambda),\exppat(\sigma,\lambda)\bigr)
\]
be the guessed exceptional data. We group these per-line guesses into the global objects
\[
\Pi:=\bigl(\pi_\lambda\bigr)_{\lambda\in\impline},\qquad
\Gamma:=\bigl(K_\lambda,\gapfn_\lambda\bigr)_{\lambda\in\impline},\qquad
\cE:=\bigl(\indx_{{\sf exp},\lambda},\exppat_\lambda\bigr)_{\lambda\in\impline}.
\]

We may assume w.l.o.g.\ that each guessed gap datum in (G3) is \emph{internally consistent}, i.e.,
for every $\lambda\in\impline$ we have
\[
K_\lambda \;=\; \sum_{g}\gapfn_\lambda(g).
\]
(If not, we discard the guess immediately.)

For each fixed global guess $(\sigma,\Pi,\Gamma,\cE)$ we build an \ddmtcsp\ instance $\cI(\sigma,\Pi,\Gamma,\cE)$,
and accept iff at least one such instance is satisfiable.

\subsubsection*{CSP Variables (global rectangle coordinates)}
We introduce four \ddmtcsp\ variables $\ell_i,r_i,b_i,t_i$ for each rectangle $R_i$.
Their domains are subsets of the discretized grids:
\[
D(\ell_i),D(r_i)\subseteq \gridpts_X,
\qquad
D(b_i),D(t_i)\subseteq \gridpts_Y.
\]
An assignment to these variables specifies the rectangle
$R_i=[\ell_i,r_i]\times[b_i,t_i]$.


\paragraph*{Skeleton as domains.}
We enforce the skeleton by restricting domains as follows.
Let
\[
X_{\impline} := \{\,x(v)\mid v\in\impline_{\ver}\,\}
\qquad\text{and}\qquad
Y_{\impline} := \{\,y(h)\mid h\in\impline_{\hor}\,\}.
\]
Define $X_{\bot} := \gridpts_X\setminus X_{\impline}
\qquad\text{and}\qquad
Y_{\bot} := \gridpts_Y\setminus Y_{\impline}.$

For each $i\in[k]$, set
\[
\begin{aligned}
D(\ell_i) &=
\begin{cases}
\{x(\sigma(i,L))\} & \text{if }\sigma(i,L)\in\impline_{\ver},\\
X_{\bot} & \text{if }\sigma(i,L)=\bot,
\end{cases}
\\[1mm]
D(r_i) &=
\begin{cases}
\{x(\sigma(i,R))\} & \text{if }\sigma(i,R)\in\impline_{\ver},\\
X_{\bot} & \text{if }\sigma(i,R)=\bot,
\end{cases}
\\[1mm]
D(b_i) &=
\begin{cases}
\{y(\sigma(i,B))\} & \text{if }\sigma(i,B)\in\impline_{\hor},\\
Y_{\bot} & \text{if }\sigma(i,B)=\bot,
\end{cases}
\\[1mm]
D(t_i) &=
\begin{cases}
\{y(\sigma(i,T))\} & \text{if }\sigma(i,T)\in\impline_{\hor},\\
Y_{\bot} & \text{if }\sigma(i,T)=\bot.
\end{cases}
\end{aligned}
\]
Here $x(v)$ is the $x$-coordinate of a vertical line $v$ and $y(h)$ is the $y$-coordinate of a horizontal line $h$.

\subsubsection*{Proper rectangles.}
We define \emph{successor maps} on the ordered grids $\gridpts_X=\langle \alpha_1<\cdots<\alpha_{N_x}\rangle$ and
$\gridpts_Y=\langle \beta_1<\cdots<\beta_{N_y}\rangle$ as follows: $\succX$ maps each $x$-grid value to the next grid value
to its right, and $\succY$ maps each $y$-grid value to the next grid value above it. Formally,
\[
\succX(\alpha_j):=\alpha_{j+1}\quad (j\in[N_x-1]),
\qquad
\succY(\beta_\ell):=\beta_{\ell+1}\quad (\ell\in[N_y-1]).
\]
We enforce the strict inequalities $\ell_i<r_i$ and $b_i<t_i$ using $\succX$ and $\succY$.
To ensure these functions are always defined, we restricted 
\[
D(\ell_i)\subseteq\{\alpha_1,\ldots,\alpha_{N_x-1}\}
\qquad\text{and}\qquad
D(b_i)\subseteq\{\beta_1,\ldots,\beta_{N_y-1}\}.
\]
Then, for each $i\in[k]$, we add the monotone constraints
\[
r_i \ \ge\ \succX(\ell_i),
\qquad
t_i \ \ge\ \succY(b_i).
\]
Since all coordinates take values from the grids, these constraints force $\ell_i<r_i$ and $b_i<t_i$.

\subsubsection*{Per-line monotone counting primitives}

We will repeatedly need to express, using monotone functions, statements of the form:
“how many points of $P$ on a given line lie before a certain coordinate?”
For a horizontal line we count along the $x$-axis, and for a vertical line we count along the $y$-axis.

\paragraph*{Horizontal line $h\in\impline_{\hor}$:}
Let $P_h:=P\cap h$ and $t_h:=|P_h|$.
For $x\in \gridpts_X$, define
\[
A_h(x)\ :=\ \bigl|\{\,p\in P_h : x(p)<x\,\}\bigr|,
\qquad
B_h(x)\ :=\ \bigl|\{\,p\in P_h : x(p)\le x\,\}\bigr|.
\]

\paragraph*{Vertical line $v\in\impline_{\ver}$:}
Let $P_v:=P\cap v$ and $t_v:=|P_v|$.
For $y\in \gridpts_Y$, define
\[
A_v(y)\ :=\ \bigl|\{\,p\in P_v : y(p)<y\,\}\bigr|,
\qquad
B_v(y)\ :=\ \bigl|\{\,p\in P_v : y(p)\le y\,\}\bigr|.
\]
All these functions are monotone non-decreasing on their respective finite domains.

\subsubsection*{Constraints for one horizontal line $h\in\impline_{\hor}$:}
\label{subsubsec:constraints-horizontal-line}

Fix a horizontal line $h\in\impline_{\hor}$.
Recall
\[
\indx(\sigma,h)=\{\,i\in[k]\mid \sigma(i,B)=h\ \text{or}\ \sigma(i,T)=h\,\},
\qquad
m(\sigma,h)=|\indx(\sigma,h)|.
\]
We also recall the endpoint-symbol set
\[
\ptsend(\sigma,h)=\{\,i_{\sf left},i_{\sf right} : i\in\indx(\sigma,h)\,\}
\]
and the guessed endpoint ordering
\[
\pi=\pi(\sigma,h):\ \ptsend(\sigma,h)\to [\,2m(\sigma,h)\,].
\]

\paragraph*{Endpoint symbols as aliases for rectangle variables:}
Endpoint symbols are \emph{combinatorial labels} used only to describe the guessed left-to-right order $\pi(\sigma,h)$.
To talk about their $x$-coordinates inside the CSP, we reuse the existing rectangle-boundary variables via the alias map
\[
\coord_h:\ \ptsend(\sigma,h)\to \{\ell_1,r_1,\ldots,\ell_k,r_k\},
\qquad
\coord_h(i_{\sf left}) := \ell_i,
\quad
\coord_h(i_{\sf right}) := r_i.
\]
For example, if $\pi^{-1}(j)=2_{\sf left}$ and $\pi^{-1}(j+1)=5_{\sf right}$, then the order constraint
$\coord_h(\pi^{-1}(j))\le \coord_h(\pi^{-1}(j+1))$ is simply $\ell_2\le r_5$. See Figure~\ref{fig:h-gaps}. 

\paragraph*{(H1) Order constraints induced by $\pi(\sigma,h)$.}
For each $j\in[\,2m(\sigma,h)-1\,]$, add
\[
\coord_h\!\bigl(\pi^{-1}(j)\bigr)
\ \le\
\coord_h\!\bigl(\pi^{-1}(j+1)\bigr).
\]


\paragraph*{(H2) Gap constraints via minimal Dyck blocks.}
Fix the ordering $\pi$ on $\ptsend(\sigma,h)$.
Let $0=j_0<j_1<\cdots<j_q=2m(\sigma,h)$ be the complete-prefix indices that define the $q$ minimal Dyck blocks
(\Cref{def:min-complete-block,lem:unique-min-dyck-blocks}).
These block boundaries carve $h$ into exactly $q+1$ \emph{gaps}:
before the first endpoint; between the end of block $r$ and the next endpoint (for each $r\in[q-1]$); and after the last endpoint.
Our guess $\gapfn(\sigma,h):[q+1]\to\{0,1,\ldots\}$ specifies how many points of $P_h$ lie in each gap.

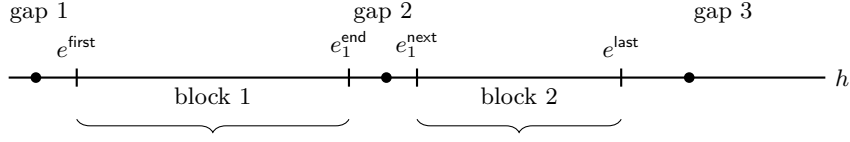
\begin{figure}[t]
\centering
\begin{tikzpicture}[x=0.9cm,y=1cm,>=Latex, font=\small]
  \draw[thick] (0,0) -- (12,0) node[right] {$h$};
  \foreach \x/\lbl in {1/{$e^{\sf first}$},5/{$e_1^{\sf end}$},6/{$e_1^{\sf next}$},9/{$e^{\sf last}$}}{
    \draw[thick] (\x,0.12) -- (\x,-0.12);
    \node[above] at (\x,0.15) {\lbl};
  }


\draw[decorate,decoration={brace,amplitude=5pt,mirror}]
  (1,-0.55) -- (5,-0.55)
  node[midway,above=2pt] {block $1$};

\draw[decorate,decoration={brace,amplitude=5pt,mirror}]
  (6,-0.55) -- (9,-0.55)
  node[midway,above=2pt] {block $2$};


  \node[above] at (0.45,0.6) {gap $1$};
  \node[above] at (5.5,0.6) {gap $2$};
  \node[above] at (10.5,0.6) {gap $3$};

  \fill (0.4,0) circle (2pt);
  \fill (5.55,0) circle (2pt);
  \fill (10.0,0) circle (2pt);

\end{tikzpicture}
\caption{ Endpoint ordering $\pi(\sigma,h)$, minimal Dyck blocks, and the induced gaps on a horizontal line:
Schematic: endpoints along $h$ (from $\pi$) decompose into minimal blocks.\\
   Gaps are the regions between blocks; $\gapfn$ prescribes how many points of $P_h$ lie in each gap.}
\label{fig:h-gaps}
\end{figure}

\smallskip
\noindent
\textbf{Formal constraints.}
Let
\[
e^{\sf first}:=\pi^{-1}(1),
\qquad
e^{\sf last}:=\pi^{-1}\!\bigl(2m(\sigma,h)\bigr),
\]
and for each $r\in[q-1]$ define
\[
e_r^{\sf end}:=\pi^{-1}(j_r),
\qquad
e_r^{\sf next}:=\pi^{-1}(j_r+1).
\]
Write $\gapfn:=\gapfn(\sigma,h)$. We add:
\begin{align*}
& A_h\!\bigl(\coord_h(e^{\sf first})\bigr) \ =\ \gapfn(1),\\
& A_h\!\bigl(\coord_h(e_r^{\sf next})\bigr) \ =\ B_h\!\bigl(\coord_h(e_r^{\sf end})\bigr) + \gapfn(r+1)
\qquad\text{for all }r\in[q-1],\\
& B_h\!\bigl(\coord_h(e^{\sf last})\bigr) \ =\ t_h - \gapfn(q+1),
\end{align*}
where $t_h:=|P_h|$.

\paragraph*{(H3) Exceptional points and the exceptional pattern.}
We now encode the \emph{exceptional} points on the horizontal line $h$ and how they are covered. Please refer to Figure~\ref{fig:h-blocks-gaps-exceptional} for an illustration. 
Recall that $\indx(\sigma,h)$ is the set of rectangles whose \emph{top} or \emph{bottom} side lies on $h$.
Every point of $P_h:=P\cap h$ that is not covered by such a horizontal side must instead be covered by a
\emph{vertical} side (left/right) of some rectangle whose sides are not aligned with $h$.
Accordingly, in the global guessing step we guess:
(i) a set $\indx_{\sf exp}(\sigma,h)\subseteq [k]\setminus \indx(\sigma,h)$ of rectangles that may cover exceptional points on $h$,
and (ii) an \emph{exceptional pattern} $\exppat(\sigma,h)$ specifying, from left to right, which rectangle and which vertical side
covers each exceptional point.

\smallskip
\noindent
\textbf{Idea.}
Let $K:=K(\sigma,h)$ be the guessed number of exceptional points on $h$.
We introduce variables $p_{h,1},\ldots,p_{h,K}$ that pick $K$ \emph{distinct} points of $P_h$ in left-to-right order.
Then we:
(a) force $p_{h,1},\ldots,p_{h,K}$ to fall into the $q+1$ gaps induced by (H2), in quantities prescribed by $\gapfn(\sigma,h)$; and
(b) enforce the guessed pattern $\exppat(\sigma,h)$ by making each chosen point coincide with the declared vertical side
($\ell_i$ or $r_i$) and ensuring this side intersects $h$.

\begin{figure}[t]
\centering
\includegraphics[width=\linewidth]{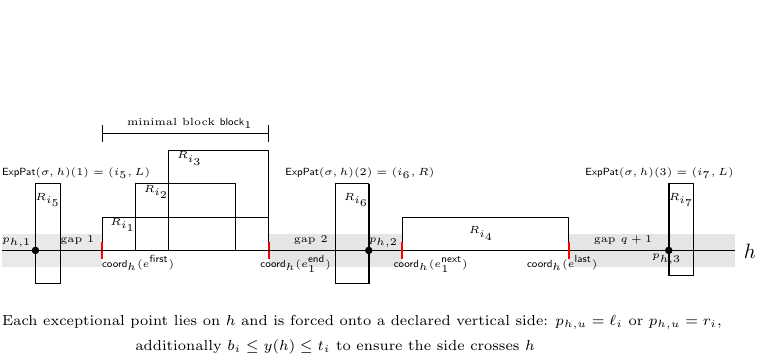}
\caption{Exceptional points on $h$: each selected point is placed inside the appropriate gap, and then forced to coincide with the declared vertical side crossing $h$ (thick side), as specified by $\exppat(\sigma,h)$.}
\label{fig:h-blocks-gaps-exceptional}
\end{figure}

\smallskip
\noindent
\textbf{Formal constraints.}
Let $\indx_{\sf exp}(\sigma,h)\subseteq [k]\setminus \indx(\sigma,h)$ and $\exppat(\sigma,h)$ be the guessed set and pattern on $h$.

\medskip
\noindent
\emph{(H3a) Choosing $K$ distinct points of $P_h$ in left-to-right order.}
Introduce variables $p_{h,1},\dots,p_{h,K}$ with domain
\[
D_h \ :=\ D(p_{h,r}) \ :=\ \{\,x(p)\mid p\in P_h\,\}\subseteq\mathbb{N}\qquad(\text{for all }r\in[K]).
\]
Let $0$ and $M_h$ be fixed constants such that $0 < \min D_h$ and $M_h > \max D_h$
(e.g.\ $M_h := \max(\gridpts_X)+1$).
Define monotone predecessor/successor maps induced by $D_h$:
\[
\predop_h(x)\ :=\ \max\bigl(\{d\in D_h : d<x\}\cup\{0\}\bigr),
\qquad
\succ_h(x)\ :=\ \min\bigl(\{d\in D_h : d>x\}\cup\{M_h\}\bigr).
\]
Enforce distinctness and left-to-right order by adding, for all $r\in[K-1]$,
\[
p_{h,r+1}\ \ge\ \succ_h(p_{h,r}).
\]
Thus $p_{h,1},\ldots,p_{h,K}$ represent $K$ distinct points of $P_h$ in increasing $x$-order.

\medskip
\noindent
\emph{(H3b) Placing the $K$ points into the $q+1$ gaps.}
Let $q$ be the number of minimal complete blocks of $\pi(\sigma,h)$ as in (H2), and
write $\gapfn:=\gapfn(\sigma,h):[q+1]\to\{0,1,\dots\}$.
Define prefix sums
\[
s_g\ :=\ \sum_{u=1}^{g}\gapfn(u)\quad(g\in[q+1]),
\qquad
s_0\ :=\ 0,
\]
so that gap $g$ corresponds to indices $r\in[s_{g-1}+1\,..\,s_g]$.
(Here $s_{q+1}=K$ by our consistency assumption on the guess.)

Recall the endpoint symbols from (H2):
\[
e^{\sf first}:=\pi^{-1}(1),
\quad
e^{\sf last}:=\pi^{-1}\!\bigl(2m(\sigma,h)\bigr),
\quad
e_r^{\sf end}:=\pi^{-1}(j_r),
\quad
e_r^{\sf next}:=\pi^{-1}(j_r+1)\ (\forall~ r\in[q-1]).
\]
We use the alias map $\coord_h$ to interpret endpoint symbols by rectangle variables.

\smallskip
\noindent

\emph{Gap $1$ (before the first endpoint).}
Since \(s_1\) is the highest index in Gap~$1$, add
\[
p_{h,s_1}\ \le\ \predop_h\!\bigl(\coord_h(e^{\sf first})\bigr).
\]

\smallskip
\noindent
\emph{Intermediate gaps.}
For each $r\in[q-1]$, in gap $r+1$, $s_r+1$ is the lowest index and $s_{r+1}$ is the highest index. Add
\[
\succ_h\!\bigl(\coord_h(e_r^{\sf end})\bigr)\ \le\ p_{h,s_r+1}\quad \text{and}\quad  p_{h,s_{r+1}} \le\ \predop_h\!\bigl(\coord_h(e_r^{\sf next})\bigr).
\]

\smallskip
\noindent
\emph{Gap $q+1$ (after the last endpoint).}
Since $s_{q}+1$ is the lowest index in the last gap, i.e, in gap $q+1$, add
\[
p_{h,s_{q}+1}\ \ge\ \succ_h\!\bigl(\coord_h(e^{\sf last})\bigr).
\]

\medskip
\noindent
\emph{(H3c) Enforcing the guessed exceptional pattern.}
Finally, enforce that each selected point is covered by the declared vertical side of the declared rectangle.
For each $u\in[K]$:
\begin{itemize}
\item if $\exppat(\sigma,h)(u)=(i,L)$, add
\[
p_{h,u}\ =\ \ell_i
\qquad\text{and}\qquad
b_i\ \le\ y(h)\ \le\ t_i;
\]
\item if $\exppat(\sigma,h)(u)=(i,R)$, add
\[
p_{h,u}\ =\ r_i
\qquad\text{and}\qquad
b_i\ \le\ y(h)\ \le\ t_i.
\]
\end{itemize}
The inequality $b_i\le y(h)\le t_i$ guarantees that the chosen vertical side of $R_i$ intersects the horizontal line $h$.
All constraints above are expressible in \ddmtcsp\ since $y(h)$ is a constant and $\predop_h,\succ_h$ are monotone.

\medskip
\noindent
\textbf{Example (informal).}
Suppose $P_h$ has $x$-coordinates $D_h=\{1,6,10\}$.
Assume the endpoint order on $h$ yields three gaps and the guess is $\gapfn(\sigma,h)=(1,1,1)$, hence $K=3$.
Then (H3a) forces $(p_{h,1},p_{h,2},p_{h,3})=(1,6,10)$ in increasing order.
The gap constraints (H3b) place $p_{h,1}$ into gap $1$, $p_{h,2}$ into the middle gap, and $p_{h,3}$ into the last gap.
Finally, (H3c) uses $\exppat(\sigma,h)$ to decide whether (say) $p_{h,2}=6$ must equal some $\ell_i$ or some $r_i$,
and $b_i\le y(h)\le t_i$ ensures that this vertical side indeed crosses $h$. 

\medskip
\noindent
\textbf{Vertical lines.}
All constraints for a vertical line $v\in\impline_{\ver}$ are obtained by the same template as for a horizontal line,
but with the roles of the two axes swapped. Concretely:
we count along the $y$-axis instead of the $x$-axis, use $A_v,B_v$ instead of $A_h,B_h$,
and swap $(\ell_i,r_i)$ with $(b_i,t_i)$. See Figure~\ref{fig:v-analogue} for an illustration. 

\paragraph*{Endpoint symbols and their coordinates on $v$.}
Here the endpoint symbols on $v$ are the bottom/top labels of rectangles whose \emph{left} or \emph{right} side lies on $v$.
Formally,
\[
\indx(\sigma,v)=\{\,i\in[k]\mid \sigma(i,L)=v\ \text{or}\ \sigma(i,R)=v\,\},\qquad
m(\sigma,v)=|\indx(\sigma,v)|,
\]
\[
\ptsend(\sigma,v)=\{\,i_{\sf bot},i_{\sf top}: i\in\indx(\sigma,v)\,\},
\qquad
\pi=\pi(\sigma,v):\ptsend(\sigma,v)\to[\,2m(\sigma,v)\,].
\]
As before, endpoint symbols are only combinatorial labels. We interpret their $y$-coordinates by aliasing them to rectangle
variables via
\[
\coord_v:\ \ptsend(\sigma,v)\to \{b_1,t_1,\ldots,b_k,t_k\},
\qquad
\coord_v(i_{\sf bot}) := b_i,
\quad
\coord_v(i_{\sf top}) := t_i.
\]
Then (V1) order constraints are:
\[
\coord_v\!\bigl(\pi^{-1}(j)\bigr)\ \le\ \coord_v\!\bigl(\pi^{-1}(j+1)\bigr)
\qquad\text{for all }j\in[\,2m(\sigma,v)-1\,].
\]

\paragraph*{Gaps, gap counts, and exceptional points on $v$.}
Using the minimal-block indices $0=j_0<\cdots<j_q=2m(\sigma,v)$ of $\pi(\sigma,v)$, we obtain $q+1$ gaps on $v$ and a guessed
$\gapfn(\sigma,v):[q+1]\to\{0,1,\dots\}$.
These counts are enforced with $A_v,B_v$ exactly as in (H2), after replacing $\coord_h$ by $\coord_v$ and $t_h$ by $t_v:=|P\cap v|$.

We introduce $K(\sigma,v)$ variables
\[
p_{v,1},\ldots,p_{v,K(\sigma,v)}\in D_v
\quad\text{where}\quad
D_v := \{\,y(p)\mid p\in P\cap v\,\},
\]
order them from bottom to top using the successor map on $D_v$, place them into the $q+1$ gaps using predecessor/successor maps on
$D_v$, and finally enforce the guessed exceptional pattern.
Since an exceptional point on a vertical line cannot be covered by rectangles in $\indx(\sigma,v)$, it must be covered by a
\emph{horizontal} side (bottom/top) of some rectangle not in $\indx(\sigma,v)$.
Accordingly, the guessed pattern is
\[
\exppat(\sigma,v):[K(\sigma,v)]\to \indx_{\sf exp}(\sigma,v)\times\{B,T\}.
\]
For each selected exceptional point $p_{v,r}$:
\begin{itemize}
\item if $\exppat(\sigma,v)(r)=(i,B)$, add
\[
p_{v,r}=b_i
\qquad\text{and}\qquad
\ell_i\ \le\ x(v)\ \le\ r_i;
\]
\item if $\exppat(\sigma,v)(r)=(i,T)$, add
\[
p_{v,r}=t_i
\qquad\text{and}\qquad
\ell_i\ \le\ x(v)\ \le\ r_i.
\]
\end{itemize}
Here $x(v)$ is the constant $x$-coordinate of the line $v$. The inequality $\ell_i\le x(v)\le r_i$ guarantees that the chosen
horizontal side of $R_i$ intersects the vertical line $v$.

\begin{figure}[t]
\centering
\begin{tikzpicture}[x=0.9cm,y=0.9cm,>=Latex, font=\small]
  \draw[thick] (0,0) -- (0,8) node[above] {$v$};
  \draw[thick] (-0.18,1) -- (0.18,1) node[right] {$\coord_v(e^{\sf first})$};
  \draw[thick] (-0.18,6.5) -- (0.18,6.5) node[right] {$\coord_v(e^{\sf last})$};

  \fill[gray!15] (-0.28,0) rectangle (0.28,1);
  \fill[gray!15] (-0.28,6.5) rectangle (0.28,8);

  \fill (0,0.55) circle (2pt) node[left=2pt] {$p_{v,1}$};
  \fill (0,7.2) circle (2pt) node[left=2pt] {$p_{v,2}$};

  \draw[thick] (0.1,0.55) -- (3.0,0.55);
  \draw[thick] (0.1,7.2) -- (2.4,7.2);

  \node[align=left,anchor=west] at (0.6,3.6)
    {Vertical-line analogue:\\
     pick $p_{v,r}$ in gaps (along $y$),\\
     then enforce $p_{v,r}=b_i$ or $t_i$ and $\ell_i\le x(v)\le r_i$.};
\end{tikzpicture}
\caption{Schematic for the vertical-line constraints: swap the roles of $x/y$ and of $(\ell_i,r_i)$ with $(b_i,t_i)$.}
\label{fig:v-analogue}
\end{figure}
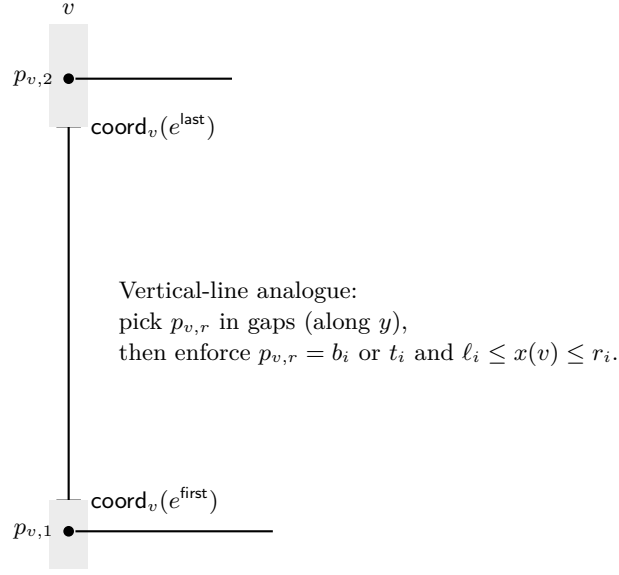

\subsection{The full CSP instance}
\label{subsubsec:full-csp-instance}

Fix a global guess $(\sigma,\Pi,\Gamma,\cE)$.
We define the \ddmtcsp\ instance $\cI(\sigma,\Pi,\Gamma,\cE)$ as the union of the following constraints and domain
restrictions:
\begin{enumerate}[label=(I\arabic*),itemsep=2pt]
\item \textbf{Global rectangle coordinates.}
For each $i\in[k]$, we include the four coordinate variables
$\ell_i,r_i\in\gridpts_X$ and $b_i,t_i\in\gridpts_Y$, together with
\begin{itemize}[itemsep=1pt]
\item the domain restrictions induced by the skeleton $\sigma$ (i.e., the constraints $D(\ell_i),D(r_i),D(b_i),D(t_i)$ as
defined in the ``Skeleton as domains'' paragraph), and
\item the proper-rectangle constraints enforcing $\ell_i<r_i$ and $b_i<t_i$ via the successor maps $\succX,\succY$.
\end{itemize}

\item \textbf{Constraints for horizontal lines.}
For every horizontal line $h\in\impline_{\hor}$ we add the constraints (H1)--(H3), namely:
\begin{itemize}[itemsep=1pt]
\item the order constraints induced by the guessed ordering $\pi(\sigma,h)$ (H1),
\item the gap-count constraints induced by the minimal-block decomposition and the guessed gap function $\gapfn(\sigma,h)$ (H2),
and
\item the exceptional-point constraints induced by $(K(\sigma,h),\indx_{\sf exp}(\sigma,h),\exppat(\sigma,h))$ (H3).
\end{itemize}

\item \textbf{Constraints for vertical lines.}
For every vertical line $v\in\impline_{\ver}$ we add the symmetric analogs of (H1)--(H3) obtained by swapping
$x\leftrightarrow y$, $(\ell_i,r_i)\leftrightarrow(b_i,t_i)$, and replacing
$A_h,B_h$ with $A_v,B_v$.
\end{enumerate}

\medskip
\noindent
\textbf{Validity as an \ddmtcsp\ instance.}
Every constraint in $\cI(\sigma,\Pi,\Gamma,\cE)$ is of the form $z \,\si\, f(z')$ with
$\si\in\{\le,\ge,=\}$ and $f$ monotone over the finite domain of $z'$:
the maps $\succX,\succY$ and the per-line predecessor/successor maps are monotone,
and the prefix-count functions $A_\lambda,B_\lambda$ are monotone nondecreasing.
Moreover, every variable domain is a finite subset of $\mathbb{N}$.
Hence $\cI(\sigma,\Pi,\Gamma,\cE)$ is a valid instance of \ddmtcsp.

\subsection{Correctness and running time}
\label{subsubsec:reduction-correctness-runtime}

We now argue correctness and analyze the running time of the reduction.

\begin{lemma}[Correctness of the reduction]
\label{lem:full-reduction-correct}
Let $(P,k)$ be an instance of \prbcshort.
Then $(P,k)$ is a \yes-instance if and only if there exists a global guess $(\sigma,\Pi,\Gamma,\cE)$ such that
the \ddmtcsp\ instance $\cI(\sigma,\Pi,\Gamma,\cE)$ is satisfiable.
\end{lemma}

\begin{proof}
We prove both directions.

\smallskip
\noindent
($\Rightarrow$)
Let $\cR=\{R_1,\dots,R_k\}$ be a grid-aligned solution that boundary-covers all points of $P$
(recall that by discretization it suffices to consider grid-aligned rectangles).
By \Cref{lem:solution-induces-skeleton}, $\cR$ induces a skeleton $\sigma$.

Fix any line $\lambda\in\impline$.
By definition of $\indx(\sigma,\lambda)$, the rectangles indexed by $\indx(\sigma,\lambda)$ are exactly those whose
side is aligned with $\lambda$ (bottom/top for $\lambda\in\impline_{\hor}$; left/right for $\lambda\in\impline_{\ver}$).
Sorting the corresponding endpoint symbols from left to right (or bottom to top) yields an ordering
$\pi(\sigma,\lambda)$, and therefore a well-defined minimal-block decomposition
and the induced gaps.
The solution $\cR$ also determines:
(i) the gap counts $\gapfn(\sigma,\lambda)$, and
(ii) the set of exceptional points on $\lambda$ together with, for each such point, which (non-aligned) rectangle side
covers it; this yields $\indx_{\sf exp}(\sigma,\lambda)$ and $\exppat(\sigma,\lambda)$.

Consider the global guess $(\sigma,\Pi,\Gamma,\cE)$ that matches these objects.
Assign the CSP variables $(\ell_i,r_i,b_i,t_i)$ to be the true coordinates of $R_i$.
For each line $\lambda$, assign the auxiliary variables $p_{\lambda,1},\ldots,p_{\lambda,K(\sigma,\lambda)}$
to be the coordinates of the exceptional points on $\lambda$ in sorted order.

We verify that all constraints of $\cI(\sigma,\Pi,\Gamma,\cE)$ are satisfied.
\begin{itemize}
\item The domain restrictions encode the skeleton, hence each rectangle boundary variable is consistent with $\sigma$.
\item The ``proper rectangle'' constraints enforce $\ell_i<r_i$ and $b_i<t_i$, which hold for every rectangle.
\item For each $\lambda\in\impline$, (H1)/(V1) holds because $\pi(\sigma,\lambda)$ is defined by sorting true endpoints.
\item For each $\lambda\in\impline$, (H2)/(V2) holds by definition of $\gapfn(\sigma,\lambda)$ and the prefix-count functions.
\item For each $\lambda\in\impline$, (H3)/(V3) holds because the chosen exceptional-point variables are placed at the
true exceptional points, which lie in the claimed gaps, and each is covered by the declared side in $\exppat(\sigma,\lambda)$
(with the side intersection inequality ensuring the side crosses $\lambda$).
\end{itemize}
Thus the CSP instance is satisfiable.

\smallskip
\noindent
($\Leftarrow$)
Let $(\sigma,\Pi,\Gamma,\cE)$ be a global guess such that $\cI(\sigma,\Pi,\Gamma,\cE)$ has a satisfying assignment.
Interpret the assigned values of $(\ell_i,r_i,b_i,t_i)$ as rectangles
$R_i=[\ell_i,r_i]\times[b_i,t_i]$ for each $i\in[k]$; by the proper-rectangle constraints they are well-defined.

Fix a horizontal line $h\in\impline_{\hor}$ (the vertical case is analogous).
We show that every point of $P_h:=P\cap h$ is boundary-covered by $\{R_1,\ldots,R_k\}$.
\begin{itemize}
\item Constraints (H2) certify that the gaps induced by the minimal-block decomposition contain exactly
$\gapfn(\sigma,h)(g)$ points each. In particular, the complement of these gaps along $h$ is exactly the union of the
endpoint intervals prescribed by the (Dyck-consistent) ordering, hence it is covered by horizontal sides of rectangles
from $\indx(\sigma,h)$ (those aligned with $h$).
\item Constraints (H3a)--(H3b) select exactly $K(\sigma,h)$ points on $h$, place them into the gaps in the prescribed
quantities, and therefore account for all points in the gaps as \emph{exceptional} points.
\item Constraints (H3c) ensure that each selected exceptional point is covered by the declared vertical side
($p_{h,u}=\ell_i$ or $p_{h,u}=r_i$) and that this side intersects $h$ via $b_i\le y(h)\le t_i$.
\end{itemize}
Hence every point of $P_h$ is boundary-covered. The same argument applies to each vertical line in $\impline_{\ver}$.

Finally, since $\impline$ covers all points of $P$ (Setup), every point of $P$ lies on some $\lambda\in\impline$, and is
therefore boundary-covered by the constructed rectangles. Thus $(P,k)$ is a \yes-instance.
\end{proof}

\begin{lemma}[Running time of the reduction]
\label{lem:full-reduction-time}
There is an algorithm that, given $(P,k)$, enumerates all global guesses $(\sigma,\Pi,\Gamma,\cE)$ and builds each
corresponding \ddmtcsp\ instance $\cI(\sigma,\Pi,\Gamma,\cE)$ in total time
$2^{\cO(k \log k)} \cdot n^{\Oh(1)}$ for some computable function $f$.
\end{lemma}

\begin{proof}
We bound (i) the number of guesses and (ii) the cost per guess.

\smallskip
\noindent
\emph{Number of global guesses.}
By \Cref{lem:number-of-skeletons} the number of skeletons $\sigma$ is $k^{\Oh(k)}$.
Fix any skeleton.
There are at most $|\impline|\le 4k$ lines.
For each line $\lambda\in\impline$, we enumerate:
\begin{itemize}
\item an endpoint ordering $\pi(\sigma,\lambda)$ over $2m(\sigma,\lambda)\le 2k$ symbols;
\item a gap vector $(K(\sigma,\lambda),\gapfn(\sigma,\lambda))$; and
\item a set $\indx_{\sf exp}(\sigma,\lambda)\subseteq [k]\setminus \indx(\sigma,\lambda)$ and a pattern $\exppat(\sigma,\lambda)$.
\end{itemize}
By the counting bounds established in the corresponding subsections (endpoint-order counting via Dyck-word structure,
and gap/exception-pattern counting), the number of combined choices per line is bounded by $k^{\Oh(k)}$. Since there are at most \(4k\) lines, the total
number of global guesses is at most \(\bigl(k^{\Oh(k)}\bigr)^{4k} = k^{\Oh(k)}\).

\smallskip
\noindent
\emph{Cost per guess.}
For each fixed guess we construct $\cI(\sigma,\Pi,\Gamma,\cE)$ as follows.
For each line $\lambda\in\impline$, we process the set $P\cap\lambda$:
sort its coordinates along the line, build arrays that evaluate $A_\lambda(\cdot)$ and $B_\lambda(\cdot)$ on the relevant
grid/domain values, and build the induced predecessor/successor maps on the finite domain
$\{x(p)\mid p\in P\cap h\}$ or $\{y(p)\mid p\in P\cap v\}$.
All of this takes polynomial time in $n$.
The total number of variables and constraints we add is bounded by a function of $k$ times a polynomial in $|P|$
(indeed, we add only $\Oh(k)$ rectangle variables and, per line, at most $\Oh(K(\sigma,\lambda)\le \Oh(k)$ auxiliary point
variables and constraints). Finally we apply Proposition~\ref{prop:distinct-domian-csp} to solve the problem in polynomial time. 

\smallskip
\noindent
Multiplying the number of guesses by the per-guess construction time yields a total running time of
$k^{\cO(k)}\cdot n^{\Oh(1)}$.
\end{proof}

Combining \Cref{lem:full-reduction-correct,lem:full-reduction-time} with the polynomial time algorithm for \ddmtcsp (Proposition~\ref{prop:distinct-domian-csp}) yields an FPT
algorithm for \prbcshort.

\paragraph*{Acknowledgements.}
We thank the anonymous reviewers for their helpful comments and suggestions on an earlier manuscript, which improved the
presentation and clarity of the paper.




\bibliography{my_bib}
\end{document}